%% file: main.tex
\documentclass[11pt]{article}

\usepackage{fullpage}
\usepackage{hdb_macros}
\usepackage{tikz}

\usepackage[multiuser,inline,nomargin]{fixme}
\fxusetheme{color}

\FXRegisterAuthor{z}{ee}{\color{cyan}Zeyong}
\FXRegisterAuthor{s}{es}{\color{red}Surendra}

\newcommand{\prob}[1]{\textsc{#1}}

\title{Hierarchies Within \cc{TFNP}: Building Blocks and Collapses}
\author{
    Surendra Ghentiyala\thanks{Cornell University. Email: \email{sg974@cornell.edu}. This work is supported in part by the NSF under Grants Nos.~CCF-2122230 and CCF-2312296, a Packard Foundation Fellowship, and a generous gift from Google.} 
    \and
    Zeyong Li\thanks{University of Toronto. Email: \email{zeyong1123@gmail.com}. Most of this work was done while the author was at the National University of Singapore, supported by NRF grant NRF-NRFI09-0005.}
}

\date{}
\begin{document}
\pagenumbering{roman}
\maketitle

\input{abstract}
\thispagestyle{empty}
\newpage

\tableofcontents
\newpage
\pagenumbering{arabic}

\section{Introduction}

\input{intro}

\section{Preliminaries}
\input{prelims}

\section{Definition}
\label{sec: def}
\input{def}
\section{Robustness of Definition}
\label{sec: def_robustness}
\input{valid_def}

\section{Self-lowness}
\subsection{\cc{PPA} Self-lowness}
\input{ppa}

\subsection{\cc{PLS} Self-lowness}
\input{pls2}

\subsection{\cc{LOSSY} Self-lowness}
\input{lossy}

\section{Further Applications}
\label{sec: consequences}

\input{consequences}
\section*{Acknowledgements}
The authors would like to thank Karthik Gajulapalli, Sidhant Saraogi, and Noah Stephens-Davidowitz
for many helpful discussions and feedback on an earlier draft of this manuscript. The authors are grateful to Mika G\"{o}\"{o}s and Valentin Imbach for helpful discussions which encouraged them to include further exposition about the utility of internal consistency and during which a bug in the $\cc{PPA}$ self-lowness proof was discovered. The authors would also like to thank the anonymous referees for useful comments.


\section*{AI Disclosure}
GPT 5.6 Sol, GPT 6, and Gemini 3.1 Pro were used to identify and correct several proof bugs. The authors (without AI) discovered that the previous proof of $\cc{PPA}$ self-lowness had a bug and devised a new proof using the same underlying idea. AI was used to fully elaborate on this proof (which AI correctly identified as incomplete). AI was used to identify and close several other important gaps. These included a bug in the $\cc{PLS}$ self-lowness proof concerning the choice of initial vertex, an observation which led to \cref{footnote: first}, and inconsistencies regarding encoding of circuits and witnesses. AI was also used for literature review and copy-editing. The authors assume full responsibility for correctness and originality.

\bibliographystyle{alpha}
\bibliography{refs}

\end{document}

%% file: abstract.tex
\begin{abstract}
    In all well-studied $\mathsf{TFNP}$ subclasses (e.g. $\mathsf{PPA}, \mathsf{PPP}$ etc.), the canonical complete problem takes as input a polynomial-size circuit $C: \{ 0, 1\}^n \rightarrow \{ 0, 1\}^m$ whose input-output behavior implicitly encodes an exponentially large object $G$, i.e. $C$ is the succinct (polynomial-size) representation of the exponential size object $G$. The goal is to find some particular substructure in $G$ which can be confirmed in polynomial time using queries to $C$. 
    
    While such formulations have proven fruitful in the $\mathsf{TFNP}$ literature, they are arguably insufficient to characterize much of $\mathsf{TFNP}$. For example, for any object $G$ whose succinct description requires one to factor integers, it seems we cannot represent it by a circuit $C$ under the widely believed assumption that $\mathrm{Factor} \notin \mathsf{FP}$.

    To address this, we initiate the study of classes of the form $\mathsf{A}^{\mathsf{B}}$ where both $\mathsf{A}$ and $\mathsf{B}$ are $\mathsf{TFNP}$ subclasses. In particular, we define complete problems for these classes that take as input a circuit $C$ which is allowed oracle gates to another $\mathsf{TFNP}$ class. For example, $\mathsf{PPP}^{\mathsf{PPA}}$ would involve finding a collision in a circuit $C^{\mathsf{PPA}}: [N] \rightarrow [N-1]$ where $C$ has oracle gates to a $\mathsf{PPA}$-complete problem. We can then iterate this construction to obtain hierarchies (e.g. $\mathsf{PPP^{(PPP^{PPP})}}$). Here, we uncover a rich structure of hierarchies and collapses. The hierarchies we obtain are contained in $\mathsf{TFNP}$ and therefore much lower than the other generalization of $\mathsf{TFNP}$ subclasses ($\mathsf{TFPH}$) recently defined in Kleinberg, Korten, Mitropolsky, and Papadimitriou (ITCS'21).

    The exact definitions require some care since, unlike a class like $\mathsf{PPP^{NP}}$, where the $\mathsf{NP}$ oracle defines a function, in $\mathsf{PPP^{PPA}}$, the oracle is for a search problem with many possible solutions. To address this, we introduce a strict condition called \emph{Internal Consistency} that underlies the definition. Informally, this condition enforces that, from the perspective of a bounded verifier, the oracle behaves like a function. 

    Beyond introducing definitions for $\mathsf{TFNP}$ oracle problems, our specific technical contributions include showing that several $\mathsf{TFNP}$ subclasses are self-low and hence their corresponding hierarchies collapse. In particular, $\mathsf{PPA^{PPA}} = \mathsf{PPA}$, $\mathsf{PLS^{PLS}} = \mathsf{PLS}$, and $\mathsf{LOSSY^{LOSSY}} = \mathsf{LOSSY}$. As an immediate consequence, we derive that when reducing to $\mathsf{PPA}$, one can always assume access to $\mathsf{PPA}$---and therefore factoring---oracle gates.
    
    In addition to introducing a variety of hierarchies within $\mathsf{TFNP}$ that merit study in their own right, these ideas introduce a novel approach for classifying computational problems within $\mathsf{TFNP}$. For example, we observe that the problem of deterministically generating large prime numbers, which has long resisted classification in a $\mathsf{TFNP}$ subclass, is in $\mathsf{PPP^{\mathsf{PPP}}}$ under the Generalized Riemann Hypothesis.
\end{abstract}



%% file: intro.tex
The complexity class $\cc{TFNP}$ (Total Functions in \NP) consists of search problems where a solution is guaranteed to exist and solutions can be verified in polynomial time \cite{Papadimitriou1994OnTC}. Although $\cc{TFNP}$ is a semantic class of problems, many syntactic subclasses of $\cc{TFNP}$ have been identified and studied, forming a rich `ecosystem' within $\cc{TFNP}$.

These syntactic subclasses typically correspond to the mathematical principle proving their totality. More specifically, the canonical complete problem for a $\cc{TFNP}$ subclass typically takes as input a polynomial-size circuit $C:\bit^n \rightarrow \bit^m$ encoding an exponentially large combinatorial object $G$. The goal is to find some particular substructure in $G$ whose existence is guaranteed by the underlying combinatorial principle. 

Take \cc{PPP} (Polynomial Pigeonhole Principle) as an illustrative example. The subclass and its canonical complete problem \prob{Pigeon} are defined as follows: 
\begin{definition}[$\prob{Pigeon}$]
    \label{def: pigeon}
    Given a $\poly(n)$ size circuit $C: \{ 0, 1\}^n \rightarrow \{ 0, 1\}^n$, output one of the following.
    \begin{enumerate}
        \item $x$ s.t. $C(x) = 0^n$;
        \item distinct $x_1, x_2$ s.t. $C(x_1) = C(x_2)$.
    \end{enumerate}
    \cc{PPP} is defined as all search problems which are many-one reducible to \prob{Pigeon}.
\end{definition}

Notice that a solution to \prob{Pigeon} always exists since if $C$ is surjective, then a type 1 solution exists, and if $C$ is not surjective, a type 2 solution exists by the pigeonhole principle. \prob{Pigeon} therefore always has a solution and is what we refer to as a total function problem. Furthermore, any solution to \prob{Pigeon} is clearly efficiently verifiable since it simply requires evaluating $C$ at most twice followed by checking equality of two polynomial length strings. 

\cc{PPP} is only one of many fundamental \cc{TFNP} subclasses; others include \cc{PPA}, \cc{PPAD}, \cc{PPADS}, and \cc{PLS}. Despite being defined with respect to circuit problems (which are interpreted as unnatural problems), these subclasses prove to be important as they capture the complexity of many natural computational problems. To name a few:
\begin{itemize}
    \item Computing a Nash equilibrium is complete for $\cc{PPAD}$ \cite{DGP09,CDT09}.
    \item \cc{PPP} captures the complexity of many cryptographic primitives \cite{SSZ18}.
    \item \cc{PLS} captures the complexity of searching for a local optimum \cite{Kre90,Sch91}.
\end{itemize}

While the identification and study of these subclasses under such formulations have greatly enhanced our understanding of $\cc{TFNP}$, we argue that existing definitions are insufficient to characterize ``most totality arguments in $\cc{TFNP}$''. For example, the well-known problems of deterministically generating large primes (which exist by the Bertrand–Chebyshev theorem) and finding a monochromatic clique in the edge coloring of an exponentially large graph (which exists by Ramsey's theorem) have yet to find a home inside $\cc{TFNP}$. (Our hierarchies will allow us to find a home for the former problem in \cref{sec: consequences}).  

In fact, a very general class of constructive principles which seem very $\cc{TFNP}$-like does not immediately have a home in any $\cc{TFNP}$ subclass. Consider the following scenario. Let $G:[2^n] \rightarrow [2^n-1]$ be a function where $G(x)$ can be evaluated in $\poly(n)$ time given access to a $\prob{Factor}$ oracle (see \cref{def: factor} for a formal definition of an oracle for factoring integers). The task is to find a collision in $G$, i.e. find $x \neq y$ such that $G(x) = G(y)$. On one hand, under the widely believed assumption that $\prob{Factor} \notin \cc{FP}$, we do not know how to evaluate $G$ using a vanilla circuit and therefore cannot reduce collision finding in $G$ to (some variant of) $\prob{Pigeon}$. On the other hand, since both $\prob{Factor}$ and $\prob{Pigeon}$ are in $\cc{TFNP}$, this task intuitively ``belongs to'' $\cc{TFNP}$. It remains total since the pigeonhole principle tells us that $G$ has a collision, regardless of whether it is efficiently computable or not. Moreover, one can efficiently verify a collision in $G$ given $x$, $y$, and the solutions to the $\prob{Factor}$ oracle calls. (The observant reader may notice that this definition is more subtle than it appears upon first inspection since the output of a $\prob{Factor}$ gate is not well-defined for integers with multiple prime factors, but we will deal with this appropriately in \cref{sec: def}.)

This raises an unsettling possibility: there may exist rather natural problems in $\cc{TFNP}$ that we do not have the right tools to define, let alone study. A natural solution to this issue is to allow circuits with oracle gates to $\cc{TFNP}$ problems as input. Consider the above example: $G$ can be succinctly represented by a polynomial-size oracle circuit $C^{\prob{Factor}}:\bit^{n} \rightarrow [2^n - 1]$ where the circuit is allowed oracle gates solving $\prob{Factor}$. One could now interpret this circuit as an input and ask for a collision!

While oracle separations between \cc{TFNP} subclasses have been well-studied \cite{beame1995relative, moriokaclassification, buresh2004relativized, goos2024separations, jain2024pigeonhole}, the idea of raising \cc{TFNP} subclasses to oracles and studying the resulting complexity classes in their own right is relatively new. \cite{kleinberg2021total} introduced $\cc{TFNP^{\Sigma_i^P}}$ as a total function analogue of \cc{PH}. They also considered complexity classes like \cc{PPP^{\Sigma_i^P}}, where the input circuit $C$ in \cref{def: pigeon} is allowed to have $\cc{\Sigma_i^P}$ oracle gates.


Although $\cc{TFNP}$ has only recently begun to be studied through the lens of oracles, oracles are ubiquitous in other areas of complexity theory. For example, the polynomial hierarchy \cc{PH}, which can be defined in terms of oracles ($\cc{\Sigma_0^P} = \cc{P}, \cc{\Sigma_{i+1}^P} = \cc{NP^{\Sigma_i^P}}$), is related to a wide variety of complexity classes and important problems. 
Just to name a few, the Sipser-Lautemann theorem \cite{lautemann1983bpp} states that $\cc{BPP} \subseteq \cc{\Sigma_2^P} \cap \cc{\Pi_2^P}$, and the circuit minimization problem is known to be in \cc{\Sigma_2^P} but not known to be in \cc{NP}. Furthermore, the assumption that the polynomial hierarchy does not collapse---that $\cc{\Sigma_i^P} \neq \cc{\Sigma_{i+1}^P}$ for any integer $\cc{i}>0$---is now a standard assumption in complexity theory.

In this work, we introduce \cc{TFNP} subclasses that have oracle access to some \cc{TFNP} problem. For example, \cc{PPP^{PPP}} would have as its complete problem $\prob{Pigeon}^\prob{Pigeon}$ which is the same as \prob{Pigeon} except that the input circuit $C$ is allowed to have \prob{Pigeon} oracle gates. While the intuition is straightforward, the actual definition requires some care since, unlike a $\cc{\Sigma_i^P}$ oracle gate, which encodes a decision problem with exactly one output, a \prob{Pigeon} instance may have an exponential number of possible outputs since $\prob{Pigeon}$ defines a relation rather than a function. The behavior of the \prob{Pigeon} oracle gate is therefore underspecified. Hence, it is somewhat unclear exactly how one should interpret the oracle gates unless the oracle problem is a search problem with unique solutions (like finding \emph{all} prime factors of an integer). Informally, we resolve this issue by quantifying over all possible functions consistent with the relation defined by the oracle problem, and requiring the oracle answers exposed during verification to be consistent with one another. \znote{added:}We call this condition \emph{Internal Consistency}. See \cref{sec: def} for a detailed discussion. 

Under the definition above, one can naturally define hierarchies of subclasses within $\TFNP$, such as the $\cc{PPA}$ hierarchy defined as $\cc{PPA}^1 = \cc{PPA}, \cc{PPA}^i = \cc{PPA}^{\cc{PPA}^{i-1}}$ for $i > 0$, and ask about the complexity of these hierarchies.
We further show that, under our definition, several classic \cc{TFNP} subclasses are self-low. That is, their corresponding hierarchies collapse to the first level.

Our ideas explore a new dimension in the complexity landscape within $\cc{TFNP}$. This opens up the exciting possibility of analogies to $\cc{PH}$. Our hierarchies hint at hitherto unasked questions about the structure of the most well-studied $\cc{TFNP}$ subclasses: Which subclass is self-low? How do the different subclasses interact with each other when given as an oracle? Are there Karp-Lipton type collapses within the $\cc{TFNP}$ hierarchies? What should standard assumptions look like in the $\cc{TFNP}$ hierarchies world (what is the equivalent of the assumption that $\PH$ does not collapse)?

\subsection{Our Contributions} 

\paragraph{Definitions.} Our first contribution is a robust definition for a new family of subclasses in \cc{TFNP}. We propose a definition of problems taking the form of $\prob{A}^{\prob{B}}$ when $\prob{A}$ is a \cc{TFNP} circuit problem (\cref{def: circuit_prob}) and \prob{B} is any \cc{TFNP} problem (\cref{def: A^B}, \cref{def: A^B_classes}). Informally, a circuit problem is one where the input consists of a circuit $C$ and some other input $a$, and an answer to the problem can be verified using black-box queries to $C$. $\prob{A}^{\prob{B}}$ takes as input a $\poly(n)$ size circuit $C^{\prob{B}}$ with oracle gates for $\prob{B}$ and some other input $a$, and outputs a solution $y, w_{1, 1}, \dots, w_{\poly(n), \poly(n)}$. Informally, one should view $y$ as a solution to $\prob{A}$ on input $(C^{\prob{B}}, a)$. However, verifying $y$ as a solution requires evaluating $C^{\prob{B}}$, which requires evaluating $\prob{B}$ oracle gates. This is where $w_{1, 1}, \dots, w_{\poly(n), \poly(n)}$ come in. We use $w_{i, j}$ as the solution to the $j^{\text{th}}$ $\prob{B}$ oracle gate query on the $i^{\text{th}}$ time the verifier for $\prob{A}$ makes a query to $C^{\prob{B}}$. In some sense, the solution to $\prob{A}^{\prob{B}}$ on input $(C^{\prob{B}}, a)$ has the same form as the solution to $\prob{A}$ but also includes the auxiliary information required to evaluate $C^{\prob{B}}$ to verify a solution.

\paragraph{Robustness of Definition.} While such a definition appears straightforward on a high level, it is indeed delicate due to the fact that a $\cc{TFNP}$ problem defines a relation and not a function. In order to obtain a robust definition, we introduce an additional property called \emph{Internal Consistency}, which intuitively enforces the oracle behave like a function from the perspective of the (polynomially bounded) verifier. 

\znote{slightly modified:}
In particular, we showed that (under the stricter internally consistent requirement), this definition has several desirable properties, making it a robust choice for defining such classes. Assume $\prob{A}$ is a \cc{TFNP} circuit problem (\cref{def: circuit_prob}) and \prob{B} is any \cc{TFNP} problem, and \cc{A} and \cc{B} are the sets of total search problems reducible to \prob{A} and \prob{B} respectively. By defining $\cc{A^B}$ as the set of total search problems reducible to $\prob{A}^{\prob{B}}$ under many-one reductions, we observe that $\cc{A^B}$ is in \cc{TFNP} (\cref{obv: A^B_in_TFNP}) and that \cc{A^B} is robust to the choice of complete problem for \cc{A} or \cc{B} (\cref{thm: A_to_A'}, \cref{thm: B_to_B'}) up to $\cc{FP^B}$ reductions.

Having defined $\cc{TFNP}$ oracle classes, we can now define whole hierarchies.
\begin{restatable*}{definition}{hierarchy}
    \label{def: hierarchy}
    Let $\cc{A}$ be a $\cc{TFNP}$ subclass and $\prob{A}$ be the canonical $\cc{A}$-complete circuit problem. We define $\prob{A}^1 = \prob{A}$, $\prob{A}^i = \prob{A}^{(\prob{A}^{i-1})}$, and $\cc{A}^i$ as all problems which have a many-one reduction to $\prob{A}^i$. The $\cc{A}$ hierarchy is defined as follows.
    \[ \cc{A}^* = \bigcup_{i \in \mathbb{N}} \cc{A}^i \]
\end{restatable*}

\paragraph{Main Theorems.} Our second contribution is a series of structural results:
\begin{restatable*}{theorem}{ppa}
\label{thm: ppa_self_low}
    $\cc{PPA}^* = \cc{PPA}$.
\end{restatable*}

\begin{restatable*}{theorem}{pls}
\label{thm: pls_self_low}
    $\cc{PLS}^* = \cc{PLS}$.
\end{restatable*}

\begin{restatable*}{theorem}{lossy}
\label{thm: lossy_self_low}
    $\cc{LOSSY}^* = \cc{LOSSY}$.
\end{restatable*}
Besides defining hierarchies within $\cc{TFNP}$, we also study the structural properties of these new hierarchies. To our surprise, we find that---unlike the polynomial hierarchy--- the $\cc{PPA}$, $\cc{PLS}$, and $\cc{LOSSY}$ hierarchies all collapse to their first levels. 
While these collapses might be perceived as straightforward conceptually (e.g. the conceptual idea behind $\cc{PLS}^\cc{PLS} = \cc{PLS}$ can be found in the language of bounded arithmetic in \cite{BussKrajicek94}), the actual proof requires great care (again due to the fact that a $\cc{TFNP}$ problem defines a relation and not a function). In fact, each of these three classes requires different technical treatments that appeal to the unique properties of each class. \znote{slightly updated above.}

One consequence of our collapses is that, e.g., when reducing a problem to any of these classes, we are free to assume oracle access to any problem in that same class. This indicates that these three classes could be more powerful than previously imagined. These results also stand in stark contrast to the fact that $\cc{PPP}$ is not even Turing-closed under black-box reductions \cite{fleming2024black}, a seemingly much weaker property than self-lowness.

\paragraph{Further Applications.} In \cref{sec: consequences}, we show how to apply our new \cc{TFNP} subclasses to study the complexity of well-known problems. One of the consequences of the fact that \cc{PPA} is self-low and that factoring is (likely) in \cc{PPA} means that one can generally assume access to a factoring oracle when reducing a problem to a \cc{PPA}-complete problem (assuming the Generalized Riemann Hypothesis). As an instantiation of this observation, we show \cref{thm: bertrand_win_win}: Let $\prob{Factor}$ be the problem of finding a non-trivial prime factor of an integer (or declaring none exist) and let \prob{Weak-Bertrand} be the problem of generating a prime between $2^n$ and $2^{32n}$ given $1^n$ as input. The following two theorems may provide a new way of attacking the longstanding open problem of pinpointing the complexity of \prob{Weak-Bertrand} in \cc{TFNP}.
\begin{restatable*}{theorem}{primesPPA}
    \label{thm: bertrand_win_win}
    If \prob{Weak-Bertrand} is in $\cc{PPA}^{\prob{Factor}}$, then the Generalized Riemann Hypothesis implies that \prob{Weak-Bertrand} is in $\cc{PPA}$.
\end{restatable*}

\begin{restatable*}{theorem}{primes}
    \label{thm: primes_in_4}
    Under the Generalized Riemann Hypothesis, \prob{Weak-Bertrand} is in $\cc{LOSSY^{PPA}}$, $\cc{LOSSY^{PPP}}$, $\cc{PPADS^{PPA}}$, $\cc{PPADS^{PPP}}$, and $\cc{PPP^{PPP}}$.
\end{restatable*}


\znote{new:}
\subsection{Internal Consistency: Technical Nuisance or Structural Property?}

\paragraph{The Motivation.} As mentioned earlier, the fact that a $\TFNP$ problem defines a relation rather than a function introduces complications in defining $\cc{TFNP}$ oracle classes. This is because the proof of totality of many classic $\TFNP$ subclasses relies on the circuit being a function! For example, it only makes sense to instantiate the pigeonhole principle in $\prob{Pigeon}$ if the underlying circuit defines a function. Otherwise one would need to contend with quirky statements such as $C^\prob{B}(x) \neq C^\prob{B}(x)$.

As such, we wish $C^\prob{B}$ to define a function (at least from the perspective of the verifier $V'$). There are at least two ways to enforce this: One is to ensure that the evaluation of the outer circuit (by the verifier $V'$) outputs the same answer for the same input, regardless of the answers to the oracle queries. We call this property \emph{External Consistency}. The other, which we use in this work, is to ensure that the answers to the oracle queries are consistent with one another. We call this property \emph{Internal Consistency}. 

It is not hard to see that Internal Consistency implies External Consistency, and it is therefore a stricter definition. One may wonder if this stricter definition is necessary, and if adopting this stricter definition buys us anything useful.

\paragraph{The Benefits.}
One major benefit of the stricter definition is that it is robust to the choice of complete problem for $\cc{A}$ or $\cc{B}$ (up to $\cc{FP^B}$ reductions). In particular, if there is a black-box reduction from $\prob{A}$ to $\prob{A}'$ and a  reduction from $\prob{B}$ to $\prob{B}'$, then there is a black-box reduction from $\prob{A}^{\prob{B}}$ to ${\prob{A}'}^{\prob{B}'}$. This is not true under the weaker definition of External Consistency, and we provide an explicit counterexample in \cref{sec:example}. We believe that such robustness is desirable and serves as a strong argument for adopting the stricter definition.

Moreover, since our principal results are upper bounds/containments (instead of lower bounds/separations), proving them under the stricter internally consistent definition yields a stronger result. 

\paragraph{The Downsides.} 
The stronger requirement also has a genuine cost. As an example, we cannot directly apply Turing-closure of $\cc{B}$ and replace adaptive \prob{B} queries by a single \prob{B} query: two such queries may each be valid in isolation yet decode to different answers on a shared original query, violating internal consistency. Preserving internal consistency therefore requires more careful treatment of the oracle queries. 

\paragraph{Structural Implications.}
One may still ask if the whole internal consistency requirement is a technical nuisance, and if there is a cleaner way to define $\cc{TFNP}$ oracle classes that retains all the desirable benefits and avoids the downsides. 

We do not have a definitive answer to this question, but in the quest of proving self-lowness of $\cc{PPA}$, $\cc{PLS}$, and $\cc{LOSSY}$ under the internally consistent definition, we encountered an interesting phenomenon where three collapse proofs enforce internal consistency in rather different ways, using the involutive structure of $\cc{PPA}$, the monotone structure in $\cc{PLS}$, and the compression/decompression structure of $\cc{LOSSY}$. (We encourage the reader to look at the proofs.)

It is somewhat surprising that enforcing internal consistency in different $\TFNP$ subclasses requires leveraging different properties of the underlying mathematical principles. We view this diversity as motivation for further study of whether internal consistency reflects a broader structural feature of $\TFNP$ subclasses, which could potentially lead to a better understanding of the structure of $\TFNP$.

\subsection{Utility of $\cc{TFNP}$ Hierarchies}

We now discuss how our proposed theory of hierarchies fits into the existing $\cc{TFNP}$ literature and how it may act as a framework to further our understanding of $\cc{TFNP}$.

\paragraph{Power of \cc{TFNP}}

Our understanding of $\cc{TFNP}$ was previously relatively flat: focused on containments/separations/intersections between subclasses. We ask ourselves about the new dimension of hierarchies and primarily explore collapses and self-lowness in this work. As noted above, we hope that future work will use our definitions to explore topics beyond self-lowness. In particular, one can imagine asking if analogous theorems for other hierarchies (e.g. $\cc{PH}$) hold in the $\cc{TFNP}$ world.

On a philosophical note, hierarchies within $\cc{TFNP}$ differ from $\cc{PH}$ in at least one important respect from our perspective. Although the notion of $\cc{TFNP^{TFNP}}$ does not mean much formally as $\cc{TFNP}$ is not believed to have complete problems, our work can be taken to say something along the lines of $\cc{TFNP^{TFNP}} = \cc{TFNP}$ (this is very much believed not to be the case for $\cc{NP}$, which forms the base of $\cc{PH}$). In particular, we see in \cref{sec: def} that oracle access to a $\cc{TFNP}$ problem does not give you any power beyond $\cc{TFNP}$ since solutions to oracle queries always have witnesses (which is not true for $\cc{NP}$).

One could draw an analogy from our self-lowness results to Turing-closure, a concept closely related to self-lowness. We say that a \cc{TFNP} subclass $\cc{A}$ is Turing-closed if $\cc{FP}^{\cc{A}} = \cc{A}$, or equivalently, the existence of a Turing reduction to \cc{A} implies the existence of a many-one reduction to \cc{A}. It is not hard to see that for any nontrivial class $\cc{A}$, self-lowness ($\cc{A^A} = \cc{A}$) implicitly requires $\cc{A}$ to be Turing-closed in the first place, since simply evaluating the input circuit to $\cc{A^A}$ is a $\cc{FP}^{\cc{A}}$ problem. In other words, self-lowness is a stronger property than Turing-closure under our definition (see \cref{obv: ppp_not_selflow}) and the self-lowness of a class indicates that it has much more power than previously believed (more so than Turing-closure).

The Turing-closure of \cc{PLS}, \cc{PPA}, \cc{PPAD}, and \cc{PPADS} was shown in \cite{buss2012propositional}. \cc{LOSSY} was shown to be Turing-closed in \cite{li2024distinguishing}. On the other hand, \cite{fleming2024black} showed that \cc{PPP} is not Turing-closed under black-box reductions. 


\paragraph{Classification of problems and upper bounds}
We are now equipped with a more powerful tool for classifying computational problems into subclasses of $\cc{TFNP}$. In particular, we are essentially free to use any $\cc{TFNP}$ subclass as a subroutine -- just invoke the corresponding oracle. Moreover, the self-lowness of, e.g. $\cc{PLS}$, tells us that if we are reducing to $\cc{PLS}$, a $\cc{PLS}$ oracle is truly free.

\cite{korten2022derandomization} considered the problem \prob{Lossy} with access to an \prob{MCSP} (minimum circuit size problem) oracle and showed that solving $\prob{Lossy}^{\prob{MCSP}}$ allows one to generate hard truth tables. This is closely related to our work, except \prob{MCSP} is not a \cc{TFNP} problem. We note that since $\prob{MCSP}$ is a decision problem, the reduction of \cite{korten2022derandomization} to $\prob{Lossy}^{\prob{MCSP}}$ does not need our definitions from \cref{sec: def}. \cite{korten2022derandomization} also considered \prob{Lossy} with a factoring oracle. They showed that if one can solve $\prob{Lossy}^{\prob{AllFactor}}$ (where $\prob{AllFactor}$ is the problem of outputting the full prime factorization of a number), then one can deterministically generate large primes, a longstanding open problem. Here again, the setting of \cite{korten2022derandomization} is more straightforward than our setting since one can simply assume an oracle gate to a problem with unique solutions. They did not consider the case where the search problem associated to the oracle is simply in \cc{TFNP}.

\paragraph{Separations and lower bounds}
We introduce new techniques for proving separations and lower bounds between subclasses of $\cc{TFNP}$:
\begin{itemize}
    \item If some subclasses form a true hierarchy, this gives us even more evidence that the class does not equal $\cc{FP}$ since this would collapse the hierarchy, much in the style of $\cc{PH}$.
    \item 
    \snote{Removed this line: One could leverage properties of both $\cc{A}$ and $\cc{B}$ to separate the class $\cc{A^B}$ from other classes.}
    \cite{li2024metamathematicsresolutionlowerbounds} studied the $\cc{TFNP}$ class rwPHP(\cc{PLS}) in the decision tree model, which consists of problems reducible to the retraction weak pigeonhole principle where the retraction function is in \cc{PLS}. To a certain extent one could interpret this class as $\cc{LOSSY}^\cc{PLS}$ in our language, but in the decision-tree model, or in the fully black-box setting. \cite{li2024metamathematicsresolutionlowerbounds} showed that this class captures the problem of proving certain restricted lower bounds for Resolution proofs.
\end{itemize}

We do note that most separation results were proved in the decision-tree model, and our definition might not directly apply there. But we are hopeful that useful ideas could be borrowed between different models.

\subsection{Open Questions}
\begin{enumerate}
    \item
    \label{item: self-low-question}
    Are \cc{PPAD}, \cc{PPADS}, \cc{CLS}, \cc{UEOPL}, and $\cc{PPA}_k$ (for $k \geq 3$) self-low?
    
    \item The newly defined complexity class \cc{PLC} (polynomial long choice) is meant to capture the combinatorial principle of the iterated pigeonhole principle \cite{pasarkar2022extremal}. \cite{pasarkar2022extremal} ask whether $\cc{PLC} \subseteq \cc{FP^{PPP}}$. The framework of hierarchies within $\cc{TFNP}$ allows us to ask more general questions. Is $\cc{PLC}$ contained in $\cc{PPP^{PPP}}$, is $\cc{PPP^{PPP}}$ contained in $\cc{PLC}$, or are the two classes incomparable?
    
\end{enumerate}

%% file: prelims.tex
\subsection{\cc{TFNP}}
We begin by formally defining a search problem and \cc{TFNP} search problems.
\begin{definition}
    A search problem is a binary relation $\mathcal{R} \subseteq \{ 0, 1\}^* \times \{ 0, 1\}^*$ where we say that $y$ is a solution to $x$ iff $(x, y) \in \mathcal{R}$.
\end{definition}

\begin{definition}[\cc{TFNP}]
    A total $\NP$ search (\cc{TFNP}) problem is a relation $\mathcal{R} \subseteq \{0, 1\}^* \times \{0, 1\}^*$ such that the following properties hold.
\begin{itemize}
    \item Polynomial: For all $(x, y) \in \mathcal{R}$, $|y| \leq \poly(|x|)$.
    \item Totality: For all inputs $x$, there is a solution $o$ such that $(x, o) \in \mathcal{R}$.
    \item \cc{FNP} membership: There exists a $\poly(|x|, |o|)$ time algorithm $V$ such that $V(x, o) = 1$ if and only if $(x, o) \in \mathcal{R}$.
\end{itemize}
\end{definition}

When dealing with \cc{TFNP} problems, we are generally concerned with many-one reductions. One should think of these as reductions with a single oracle call.
\begin{definition}
    \label{def: reduction}
    Let $\mathcal{R}$ and $\mathcal{Q}$ be search relations. A many-one reduction from $\mathcal{R}$ to $\mathcal{Q}$ consists of two polynomial-time algorithms $f$ and $g$ such that, for every $x, y \in \{0,1\}^*$,
    \[ (f(x), y) \in \mathcal{Q} \implies (x, g(x,y)) \in \mathcal{R}. \]
\end{definition}

Alternatively, there is the notion of a Turing reduction, where one can make multiple oracle calls. Informally, we say that a class is Turing-closed if Turing reductions give us no more power than many-one reductions.
\begin{definition}[\cite{fleming2024black}]
    We say that a search problem $\mathcal{R}$ is Turing-closed if any problem which is polynomial time reducible to $\mathcal{R}$ via multiple calls to an oracle for a problem in $\mathcal{R}$ is also polynomial time reducible to $\mathcal{R}$ using a single call to an oracle for a problem in $\mathcal{R}$. We say a complexity class with a complete problem is Turing-closed if its complete problem is Turing closed.
\end{definition}

Quite crucially for us, a variety of important \cc{TFNP} subclasses are known to be Turing-closed.
\begin{lemma}[\cite{buss2012propositional}]
    \label{lem: turing-closed}
    \cc{PPA}, \cc{PPAD}, \cc{PPADS}, and \cc{PLS} are Turing-closed.
\end{lemma}

\begin{lemma}[\cite{li2024distinguishing}]
    \label{lem: lossy-turing-closed}
    \cc{LOSSY} is Turing-closed.
\end{lemma}


\subsection{Some \cc{TFNP} subclasses}
\label{subsec: problems}
Here we review some \cc{TFNP} subclasses and give some informal intuition regarding their structures. All of these classes of interest will involve a polynomial size circuit implicitly encoding an exponential size object. The goal will be to find (an efficiently verifiable) structure which must exist in this exponential size object. We note that we will freely switch between a binary string or set ($\{0, 1\}^n$) and its integer representation ($[2^n]$). We encourage the reader to not be overly concerned with this technical detail. 

\snote{Changed def of $\cc{PPA}$. Let's just use $\prob{Lonely}$ rather than the convoluted $\prob{Bipartite-mod-k}$.}
We begin by defining $\cc{PPA}$ with respect to the complete problem $\prob{Lonely}$ \cite{beame1995relative}.
\begin{definition}
    The problem \prob{Lonely} is defined as follows. The input is a circuit $C: [N] \rightarrow [N]$ such that $N$ is odd (which may be enforced syntactically). Output $w \in [N]$ such that either $C(w) = w$ or $C(C(w)) \neq w$. $\cc{PPA}$ is defined as all $\cc{TFNP}$ problems which are many-one reducible to the problem $\prob{Lonely}$.
\end{definition}
In the $\prob{Lonely}$ problem, $C$ implicitly encodes a pairing of vertices $[N]$ where $C(x) = y$ means that $x$ and $y$ are paired. The goal is to output either an unpaired vertex ($w$ such that $C(w) = w$) or a witness that $C$ does not encode a pairing ($w$ such that $C(C(w)) \neq w$). \prob{Lonely}'s totality follows from the fact that if $N$ is odd, then no perfect pairing of all $N$ vertices can exist. We will not do so here, but one can define the more general family of classes $\cc{PPA}_k$ (where $\cc{PPA}$ is equivalent to $\cc{PPA}_2$) via the complete problem $\prob{Bipartite-mod-k}$ \cite{Papadimitriou1994OnTC, hollender2021classes}.

\begin{definition}[\cc{PPAD} and \prob{EndOfLine}]
    \label{def: ppad}
     The problem \prob{EndOfLine} is defined as follows. Given $S: \{ 0, 1\}^n \rightarrow \{ 0, 1\}^n, P: \{ 0, 1\}^n \rightarrow \{ 0, 1\}^n$ such that $S(0) \neq 0$ and $P(0) = 0$, output $x$ such that $P(S(x)) \neq x$ or $x \neq 0$ such that $S(P(x)) \neq x$. \cc{PPAD} is defined as all \cc{TFNP} problems which are many-one reducible to \prob{EndOfLine}.
\end{definition}
The circuit for \prob{EndOfLine} should be viewed as specifying a directed graph. The input gives us a successor circuit $S$ and a predecessor circuit $P$. We say that the graph implicitly defined by $S, P$ has an edge from $x$ to $y$ if $S(x) = y$ and $P(y) = x$. Notice that there is no edge leading to $0$ in any such graph since $P(0) = 0$. Therefore, there must be a node $x$ which has no outgoing edges, which implies $P(S(x)) \neq x$. This can be thought of as a sink of a line in the graph defined by $S, P$. Notice that \prob{EndOfLine} also allows for a solution $x$ such that $S(P(x)) \neq x$. This should be interpreted as the beginning of a new line in the graph implicitly encoded by $S, P$ (one other than the one starting at $0$), a node which has an edge out but no incoming edges.

\begin{definition}[\cc{PPADS} and \prob{SinkOfLine}]
    \label{def: ppads}
     The problem \prob{SinkOfLine} is defined as follows. Given $S: \{ 0, 1\}^n \rightarrow \{ 0, 1\}^n, P: \{ 0, 1\}^n \rightarrow \{ 0, 1\}^n$ such that $S(0) \neq 0$ and $P(0) = 0$, output $x$ such that $P(S(x)) \neq x$. \cc{PPADS} is defined as all \cc{TFNP} problems which are many-one reducible to \prob{SinkOfLine}.
\end{definition}
\prob{SinkOfLine} is almost the same as \prob{EndOfLine} except that we only allow one type of solution: a sink in the graph implicitly defined by $S, P$. Beginnings of a new line are no longer solutions. One can also define the following useful \cc{PPADS}-complete problem (under blackbox reductions) which we will use in \cref{sec: consequences}.
\begin{definition}
    \label{def: inj_pigeon}
    The problem \prob{Injective-Pigeon} is defined as follows. Given $C: [2^n] \rightarrow [2^n-1], D: [2^n-1] \rightarrow [2^n]$, output $x$ such that $D(C(x)) \neq x$.
\end{definition}

\begin{definition}[\cc{PLS} and \prob{Sink-of-DAG}]
    The \prob{Sink-of-DAG} problem is defined as follows. Given $\poly(n)$ size circuits $S: [2^n] \rightarrow [2^n]$ and $V: [2^n] \rightarrow [2^n]$ such that $S(0) \neq 0$, find $v$ such that $S(v) \neq v$ and either $S(S(v)) = S(v)$ or $V(S(v)) \leq V(v)$. \cc{PLS} is the set of all \cc{TFNP} problems which are many-one reducible to \prob{Sink-of-DAG}.
\end{definition}
\prob{Sink-of-DAG} should be viewed as encoding a local search problem. There are two circuits, a successor circuit and a value circuit. At every point $v$ in the space, we hope that the successor function $S$ leads us to a (different) point which has a higher value $(V(S(v)) > V(v))$. A solution to \prob{Sink-of-DAG} is a point such that this condition is violated ($S(v) \neq v$ but $V(S(v)) \leq V(v)$), or one which acts as a sink in the gradient ascent process ($S(v) \neq v$ but $S(S(v)) = S(v)$). 

\begin{definition}[\cc{PWPP} and \prob{Weak-Pigeon}]
    The \prob{Weak-Pigeon} problem is defined as follows. Given a $\poly(n)$ size circuit $C: \{ 0, 1\}^n \rightarrow \{ 0, 1\}^{n-1}$, output distinct $x_1, x_2 \in \{ 0, 1\}^n$ such that $C(x_1) = C(x_2)$. \cc{PWPP} is the set of all \cc{TFNP} problems which are many-one reducible to \prob{Weak-Pigeon}.
\end{definition}
\cc{PWPP} should be considered the algorithmic analogue of the weak pigeonhole principle. We know that a collision exists in $C$ since $C$ is compressing; \prob{Weak-Pigeon} asks us to find a collision.

\begin{definition}
    The problem $f(n)$-\prob{Lossy} is defined as follows. Given $C: \{ 0, 1\}^{n} \rightarrow \{ 0, 1\}^{f(n)}$ and $D: \{ 0, 1\}^{f(n)} \rightarrow \{ 0, 1\}^{n}$, output $x \in \{ 0, 1\}^{n}$ such that $D(C(x)) \neq x$. We refer to $(n/2)$-\prob{Lossy} simply as \prob{Lossy}. \cc{LOSSY} is defined as all \cc{TFNP} problems which are many-one reducible to \prob{Lossy}.
\end{definition}

\cite{korten2022derandomization} defined \prob{Lossy} (which they call \prob{Lossy Code}), but did not define the complexity class \cc{LOSSY}. We believe this is the correct and natural definition for \cc{LOSSY}. One should view the inputs to $f(n)$-\prob{Lossy} as consisting of a compressor circuit $C$ and a decompressor circuit $D$. The goal is to find a string that is not compressible by this compression scheme. Such a string must exist since the compression scheme is lossy. The following lemma shows that the compression factor of $C$ and $D$ does not matter (up to a polynomial factor).
\begin{lemma}[\cite{korten2022derandomization}]
    \label{lem: lossy-equiv}
    $f(n)$-\prob{Lossy} is many-one equivalent to \prob{Lossy} for any efficiently computable $f(n) < n$ and $f(n) = \poly(n)$.
\end{lemma}

\begin{definition}[\prob{Factor}]
    \label{def: factor}
    \prob{Factor} is defined as follows. Given an $n$ bit integer $x$, output $0^{n-1}$ if $x$ is prime or $x \in \{0,1\}$. Otherwise, output a non-trivial divisor $y \in \{ 0, 1\}^{n-1}$ where $y$ divides $x$.
\end{definition}
Notice that an $n$ bit composite number has an $n-1$ bit non-trivial divisor. Therefore, the size of the solution to \prob{Factor} on an $n$ bit input is bounded above by $n-1$. Furthermore, as was shown in \cite{agrawal2004primes}, testing if $x$ is prime can be done in polynomial time.

%% file: def.tex
\subsection{Oracle Gates}
We will be considering some \cc{TFNP} problem $\prob{A}$ given access to an oracle for a \cc{TFNP} problem \prob{B}. However, this notion will only make sense (at least as we define it) when \prob{A} is what we term a ``circuit problem''.
\begin{definition}
    \label{def: circuit_prob}
    We say \prob{A} is a circuit problem if the input to \prob{A} is a $\poly(n)$ size circuit $C: \{ 0, 1\}^{n} \rightarrow \{ 0, 1\}^{\poly(n)}$ and possibly some other $\poly(n)$ size input $a \in \{ 0, 1\}^{\poly(n)}$.
\end{definition}

Throughout, when a circuit problem is said to have a black-box verifier, we also require its totality to be black-box: the existence of a solution may depend only on the input-output behavior of the circuit, and not on the internal structure of the circuit. 

We note that the complete problems for all the major \cc{TFNP} classes are in fact circuit problems (\cref{subsec: problems}). There is a minor subtlety that some problems, like \prob{Sink-of-DAG}, take as input more than one circuit $S, V$. In these cases, we treat the circuits as a single circuit $SV$, where $SV(x) = S(x) \Vert V(x)$. As an example, \prob{Sink-of-DAG} is a circuit problem since it has as input circuits $S$ and $V$ (which we treat as a single circuit). We assume a canonical gate evaluation order over circuits. Since \prob{B} is a \cc{TFNP} problem, the length of any solution on input $x$ is bounded by some polynomial $p(|x|)$. We will insist that all 
\prob{B} oracle gates $G$ in $C^{\prob{B}}$ have the form $G: \{ 0, 1\}^{z} \rightarrow \{ 0, 1\}^{p(|z|)}$. This is to ensure that a valid solution always fits within the number of output wires of $G$.

\begin{example}
    \label{ex: pwpp}
    To illustrate some of the concepts around \cc{TFNP} classes with oracles, we will review concepts by instantiating them with the $\prob{Weak-Pigeon}^{\prob{Factor}}$ instance of \cref{fig: circuit}. For the remainder of this section, we refer to the circuit defined in \cref{fig: circuit} as $T: \{ 0, 1\}^6 \rightarrow \{ 0, 1\}^5$.
\end{example}

\usetikzlibrary{shapes, arrows, calc}
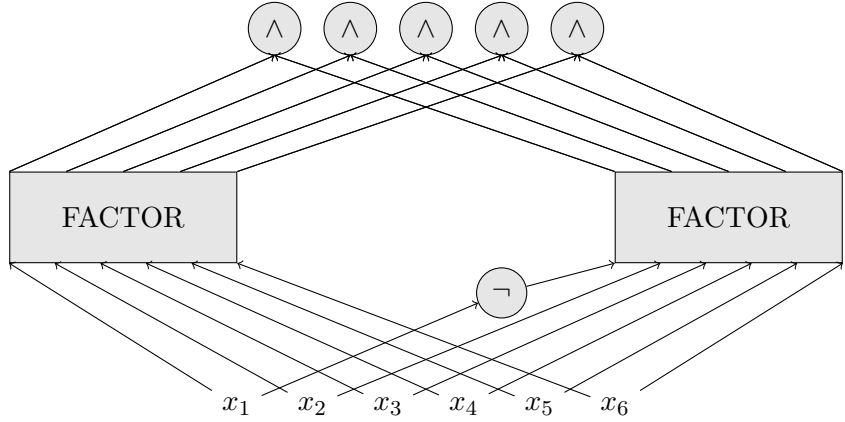
\begin{figure}
    \centering
    \begin{tikzpicture}[scale=1, every node/.style={scale=1}]
    \foreach \i in {1,2,3,4,5,6} {
        \node[] (x\i) at (\i-3.5,0) {$x_\i$};
    }

    \node[draw, fill=gray!20, minimum width=3cm, minimum height=1.2cm] (factorL) at (-4,2.5) {FACTOR};  

    \node[draw, fill=gray!20, minimum width=3cm, minimum height=1.2cm] (factorR) at (4,2.5) {FACTOR};  

    \node[circle, draw, fill=gray!20] (xor) at (1,1.5) {$\lnot$};  

    \draw [->] (x1) -- ($ (factorL.south west)!{(1-1)/5}!(factorL.south east) $);
    \foreach \i in {2,3,4,5,6} {  
        \draw [->] (x\i) -- ($ (factorL.south west)!{(\i-1)/5}!(factorL.south east) $);
        \draw [->] (x\i) -- ($ (factorR.south west)!{(\i-1)/5}!(factorR.south east) $);
    }

    \draw [->] (x1) -- (xor);
    \draw [->] (xor) -- ($ (factorR.south west)!0!(factorR.south east) $);

    \foreach \i in {1,2,3,4,5} {
        \node[circle, draw, fill=gray!20] (and\i) at (\i-3,5) {$\land$};
    }

    \foreach \i in {1,2,3,4,5} {
        \draw [->] ($ (factorL.north west)!{(\i-1)/4}!(factorL.north east) $) -- (and\i.south);
        \draw [->] ($ (factorR.north west)!{(\i-1)/4}!(factorR.north east) $) -- (and\i.south);

        \draw [->] ($ (factorL.north west)!{(\i-1)/4}!(factorL.north east) $) -- (and\i.south);
        \draw [->] ($ (factorR.north west)!{(\i-1)/4}!(factorR.north east) $) -- (and\i.south);
    }

\end{tikzpicture}

    \caption{A $\prob{Weak-Pigeon}^{\prob{Factor}}$ instance}
    \label{fig: circuit}
\end{figure}

\subsection{Evaluating Oracle Circuits with Auxiliary Oracle Answers: $C^*$}
\label{subsubsec: c^*}
We now reiterate why defining $\prob{A}^{\prob{B}}$ for \cc{TFNP} problems \prob{A} and \prob{B} is more challenging than defining $\prob{A}^{\prob{B}}$ when $\prob{B}$ is a decision problem, like \prob{SAT}. When the oracle gates in $C^{\prob{SAT}}$ are $\prob{SAT}$ gates, the behavior of the circuit on any input $x$ is well-defined since every input to an oracle gate has a unique output ($0$ or $1$). However, when the oracle gates in $C^{\prob{Pigeon}}$ are gates for a search problem (e.g., \prob{Pigeon}), the behavior of the circuit $C^\prob{Pigeon}$ on an input $x$ may not be well-defined, since the evaluation of the oracle gates could have many possible solutions. If the oracle gates of $C$ are for a problem with unique solutions (e.g., finding all prime factors of a number), this problem does not occur, but we want to work with full generality. Towards addressing this issue, we define a helper function $C^*$ which simulates $C^\prob{B}$ with auxiliary inputs that serve as solutions to its oracle gates.

\begin{definition}
    Let $C^{\prob{O}}: \{ 0, 1\}^n \rightarrow \{ 0, 1\}^m$ be an oracle circuit where $\prob{O}$ is a \cc{TFNP} problem, $C^{\prob{O}}$ contains $t$ oracle gates, and the $i^{\text{th}}$ oracle gate of $C^{\prob{O}}$ has an $s_i$-bit output. We now define
    \[ C^*: \{ 0, 1\}^n \times \prod_{i=1}^t \{ 0, 1\}^{s_i} \rightarrow \{ 0, 1\}^m \cup \{\bot\} \; .\]
    
    $C^*$ on input $(x, w_1, \ldots, w_t)$ simulates the evaluation of $C^{\prob{O}}$ on $x$. At the $i^{th}$ instance when $C^{\prob{O}}$ must evaluate an oracle gate, say on input $u$, we check if $w_i$ has a suffix that is a valid solution to \prob{O} on input $u$. If it is not, then $C^*(x, w_1, \dots, w_t) = \bot$ and the simulation terminates; otherwise, the simulation of $C^{\prob{O}}$ continues by using $w_i$ as the output of oracle gate $i$. If the simulation does not terminate prematurely and obtains a simulated $m$-bit output of $C^\prob{O}(x)$, $C^*(x, w_1, \dots, w_t)$ returns that output.
\end{definition}

$C^*(x, w_1, \dots, w_t)$ is the evaluation of $C^\prob{O}$ on $x$ given that the output of oracle gate $i$ is $w_i$. In other words, $(x, w_1, \dots, w_t)$ contains all the information necessary to lock in an evaluation of $C^\prob{O}$ on $x$.

We note that $w_i$ may not itself be a solution to the query made to oracle gate $i$ during the evaluation of $C^{\prob{O}}(x)$, but some suffix of $w_i$ is a valid solution. We do this because oracle gates have a fixed number of output wires but we still wish to allow for variable length solutions to queries to oracle gates. This is without loss of generality, as the circuit can determine for itself which suffix is a solution to the query it made to the oracle gate. Also, for the rest of the presentation, when we say $w_i$ is a solution for some oracle query $u$, we mean that some suffix $w_i'$ where $w_i = r \Vert w_i'$ is a solution to $u$. Moreover, for some circuit $C_u$ defined on $\{0,1\}^{|w_i'|}$, we abuse the notation and define $C_u(w_i) := r \Vert C_u(w_i')$ for simplicity. 

\begin{example}[Continuation of \cref{ex: pwpp}]
    \label{ex: pwpp2}
    We now provide some intuition for what $T^*$ does by probing it on a few values. We assume that the gates of $T$ are evaluated from left to right. In the following examples, we freely switch between integers and their binary representations. Consider the evaluation of $T^*(10, 2, 17)$. We begin by simulating $T$ on the input integer $10$, which is $001010$ in binary. This leads us to query $001010$ at the left \prob{Factor} gate. We see that $2$ is indeed a factor of $10$, so we continue the simulation of $T$ by assuming that the \prob{Left} oracle gate outputs $2$ ($00010$ in binary). The simulation then queries $101010$ (42 in binary) at the right \prob{Factor} gate. We see that $17$ is not a factor of 42, so $T^*(10, 2, 17) = \bot$.

    Consider the evaluation of $T^*(10, 2, 21)$. Again, we begin by simulating $T$ on the input integer $10$, which is $001010$ in binary. This leads us to query $001010$ at the left \prob{Factor} gate. We see that $2$ is indeed a factor of $10$, so we continue the simulation of $T$ by assuming that the left oracle gate outputs $2$ ($00010$ in binary). The simulation next queries $101010$ (42 in binary) at the right \prob{Factor} gate. We see that $21$ is indeed a factor of $42$, so the right \prob{Factor} gate outputs $21$ ($10101$ in binary). Therefore, $T^*(10, 2, 21) = 00010 \land 10101 = 00000$.
\end{example}

\subsection{The Full Definition}

We are now ready to define what a solution to $\prob{A}^{\prob{B}}$ looks like.
\begin{definition}
    \label{def: A^B}
    Let \prob{A} and \prob{B} be \cc{TFNP} problems and \prob{A} be a circuit problem with a black-box verifier $V: \{ 0, 1\}^{n} \times \{ 0, 1\}^{\poly(n)} \rightarrow \{ 0, 1\}$, one which only queries the circuit input to $\prob{A}$ in a black-box manner. The input to $\prob{A}^{\prob{B}}$ is defined the same as $\prob{A}$ except the circuit input $C^{\prob{B}}$ of \prob{A} contains $t$ \prob{B} oracle gates. The verifier for $\prob{A}^{\prob{B}}$, $V': \{ 0, 1\}^{n} \times \{ 0, 1\}^{\poly(n)} \rightarrow \{ 0, 1\}$ is defined as follows. $V'$ takes as input $(C^\prob{B}, a), (y, w_{1, 1}, \dots, w_{\poly(n), t})$. $V'$ simulates the computation of $V((C^\prob{B},a), y)$ and on the $i^{th}$ evaluation of $C^\prob{B}$, say on input $x$, the black-box verifier $V'$ receives $C^*(x, w_{i, 1}, w_{i, 2}, \dots, w_{i, t})$. If at any point, $C^*(x, w_{i, 1}, w_{i, 2}, \dots, w_{i, t}) = \bot$ or $V'$ finds the sub-witnesses $w_{i, j}$ to not be \emph{internally consistent} (which we define below), $V'$ outputs $0$. Otherwise, $V'$ outputs the result of $V$ after the simulation of $V$ terminates.

    We say that the $w_{1, 1}, \dots, w_{\poly(n), t}$ are not \emph{internally consistent} (with respect to the simulation) if in the execution of $V'$, there exist two oracle queries made on the same input $u$ and the simulation uses answer $w_{a, b}$ to the first oracle gate query and $w_{a', b'}$ as the answer to the second oracle gate query and $w_{a, b} \neq w_{a', b'}$.
    
\end{definition}

\znote{Discussion on internal consistency moved to intro}




In particular, a solution to $\prob{A}^{\prob{B}}$ is a solution to \prob{A} along with all the answers to the $\prob{B}$ oracle gate queries that are made when verifying such a solution. 
We allow any input/output of the \prob{B} gates as long as it is actually a solution to \prob{B} and the behavior is consistent across multiple evaluations of \prob{B}. In some sense, we are quantifying over all instantiations (i.e. all functions consistent with the relation defined by \prob{B}) of a \prob{B} oracle. Under \cref{def: A^B}, a reduction from some problem \prob{C} to $\prob{A}^{\prob{B}}$ should work for all instantiations of the oracle gates for \prob{B}. An alternative perspective is that we are delegating the work of fixing the \prob{B} oracle to the solution $y, w_{1, 1}, \dots, w_{\poly(n), t}$, particularly the $w_{1, 1}, \dots, w_{\poly(n), t}$.

\begin{example}[Continuation of \cref{ex: pwpp2}]
    \label{ex: pwpp3}
    Using \cref{def: A^B}, we see that a solution to $\prob{Weak-Pigeon}^{\prob{Factor}}$ on input $T$ would be distinct $x_1, x_2$ and internally consistent $w_1, w_2, w_3, w_4$ such that $T^*(x_1, w_1, w_2) \neq \bot$ and $T^*(x_1, w_1, w_2) = T^*(x_2, w_3, w_4)$. We will now try to find solutions to our $\prob{Weak-Pigeon}^{\prob{Factor}}$ instance $T$. One seemingly obvious solution is $(x_1, x_2) = (0, 32), (w_1, w_2) = (0, 2), (w_3, w_4) =  (8, 0)$. This appears to be a solution since $T^*(0, 0, 2) = 0$ and $T^*(32, 8, 0) = 0$. However, observe that this solution is not internally consistent. In particular, when evaluating $T^*(0, 0, 2)$, we assume that \prob{Factor} on input $32$ evaluates to 2, but when evaluating $T^*(32, 8, 0)$, we assume that \prob{Factor} on input $32$ evaluates to $8$.

    An actual solution for $\prob{Weak-Pigeon}^{\prob{Factor}}$ on instance $T$ (which we leave to the reader to confirm) is $(x_1, x_2) = (0, 32), (w_1, w_2) = (0, 2), (w_3, w_4) =  (2, 0)$.
\end{example}


\begin{definition}
    \label{def: A^B_classes}
    Let $\cc{A}, \cc{B} \subseteq \cc{TFNP}$, let \prob{A} and \prob{B} be canonical complete problems for \cc{A} and \cc{B} respectively, and let \prob{A} be a circuit problem with a black-box verifier. $\cc{A}^{\cc{B}}$ is then defined as all problems which are many-to-one reducible to $\prob{A}^{\prob{B}}$.
\end{definition}

With the definition of $\cc{A}^{\cc{B}}$ in place, the definition of hierarchies follows naturally.
\hierarchy

\subsection{Another Helpful Evaluation: $C_*$}
We now define an alternative way to evaluate $C^{\prob{O}}$ on $x$ that is significantly different from $C^*$. While $C^*$ may be easier to reason about, $C_*$ will be useful in enforcing internal consistency of our solutions.

\begin{definition}
    \label{def: another_eval}
    Let $C^{\prob{O}}: \{ 0, 1\}^n \rightarrow \{ 0, 1\}^m$ be an oracle circuit where $\prob{O}$ is a \cc{TFNP} problem, $C^{\prob{O}}$ contains $t$ oracle gates, and the $i^{\text{th}}$ oracle gate of $C^{\prob{O}}$ has an $s_i$-bit output. We now define
    \[ C_*: \{ 0, 1\}^n \times \prod_{i=1}^t \{ 0, 1\}^{s_i} \rightarrow \{ 0, 1\}^m \cup (\{0, 1\}^{\ceil{\log t}}\times \{0,1\}^{\poly(n)}) \; .\]
    
    We maintain an ongoing set $M$, initially empty. $C_*$ on input $(x, w_1, \ldots, w_t)$ simulates the evaluation of $C^{\prob{O}}$ on $x$. At the $i^{th}$ instance when $C^{\prob{O}}$ must evaluate an oracle gate $G_i$, say on input $u$, if $0$ is a valid solution to $G_i$ on $u$, we continue simulating $C^\prob{O}$ using $0$ as the output of oracle gate $G_i$. Otherwise, we find the minimum $j$ such that $w_j$ is a solution for \prob{O} on input $u$ and update $M \leftarrow M \cup \{ j \}$. We refer to this as the $i^{\text{th}}$ witness used to evaluate $C_*(x, w_1, \dots, w_t)$. If no such $j$ exists when evaluating oracle gate $i$, let $k$ be the smallest index such that $k \notin M$ and we output $(k, u)$. We refer to this as an error output and it signals where the simulation fails. Otherwise, we continue simulating $C^{\prob{O}}$ by using $w_j$ as the output of oracle gate $i$. If the simulation does not terminate prematurely and obtains a simulated $m$-bit output of $C^\prob{O}(x)$, $C_*(x, w_1, \dots, w_t)$ returns that output. We refer to this as a result output. We say $w_j$ was used in the evaluation of $C_*(x, w_1, \dots, w_t)$ if $j \in M$ when the procedure for $C_*$ terminates.
\end{definition}

The key difference between $C^*$ and $C_*$ is that instead of using the subsequent $w_{i}$ as the solution to an oracle gate query in the simulation of $C^{\prob{O}}$ as we did in $C^*$, we find the first $w_i$ which is sufficient to answer the oracle gate query and use that in the simulation of $C^{\prob{O}}$. We refer to this as evaluation by first witnesses since we use the first solution possible whenever evaluating an oracle gate. We note that such evaluation naturally enforces internal consistency. 

\begin{example}[Continuation of \cref{ex: pwpp3}]
    To understand how $T_*$ is evaluated, consider $T_*(10, 2, 21)$. We begin by simulating $T$ on the input $10$, or $001010$ in binary. We query the left \prob{Factor} gate on $10$ and find that $2$ is indeed a factor of $10$ (and in particular, the first value which is a factor among $2, 21$), so we assume that the left \prob{Factor} gate outputs $00010$. We next evaluate the right \prob{Factor} gate on $42$ and find that $2$ is a factor of $42$ (and in particular, the first value which is a factor among $2, 21$), so we continue the evaluation assuming that the right \prob{Factor} gate outputs $2$. This leads the simulation of $T$ to output $00010 \land 00010 = 00010$, or $2$. Therefore $T_*(10, 2, 21) = 2$. The observant reader will note that this is a different result from \cref{subsubsec: c^*}, where $T^*(10, 2, 21) = 0$.

    Finally, let us consider an example where $T_*$ outputs an error result. Consider $T_*(14, 7, 9)$. We query the left \prob{Factor} gate on $14$ and find that $7$ is indeed a factor of $14$ (and in particular, the first value which is a factor among $7, 9$), so we assume that the left \prob{Factor} gate outputs $00111$. We also add $1$ to our set $M$. We next evaluate the right \prob{Factor} gate on $46$ and find that neither $7$ nor $9$ is a factor of $46$. We therefore terminate and $T_*(14, 7, 9) = (2, 46)$, indicating that $w_2$ is the next unused witness. 
\end{example}

%% file: valid_def.tex
Having defined $\prob{A}^{\prob{B}}$ for \cc{TFNP} problems, we critique our definitions and show that they have several desirable properties. 

\subsection{Hierarchies Remain in $\cc{TFNP}$}
We first show that $\prob{A}^{\prob{B}}$ and its corresponding class is in \cc{TFNP}.
\begin{theorem}
    \label{obv: A^B_in_TFNP}
    If \prob{A} and \prob{B} are \cc{TFNP} problems and \prob{A} is a circuit problem with a black-box verifier, then $\prob{A}^{\prob{B}} \in \cc{TFNP}$.
\end{theorem}
\begin{proof}
    The efficiency of the verifier for $\prob{A}^{\prob{B}}$ and the fact that $\prob{A}^{\prob{B}}$ has polynomially bounded solutions are inherited from the efficiency of the verifier for $\prob{A}$. To show totality, let us consider a problem $\prob{B}'$ defined as $\{ (x, y) : \text{$y$ is the lexicographically first solution to \prob{B} on input $x$} \}$. Notice then that $\prob{A}^{\prob{B}'}$ is total since $\prob{A}$ is total. But of course, any solution to $\prob{A}^{\prob{B}'}$ is a solution to $\prob{A}^{\prob{B}}$. Therefore, $\prob{A}^{\prob{B}}$ is total.
\end{proof}

\begin{observation}
    Let $\cc{A}, \cc{B} \in \cc{TFNP}$, let \prob{A} and \prob{B} be canonical complete problems for \cc{A} and \cc{B} respectively, and let \prob{A} be a circuit problem with a black-box verifier. $\cc{A}^{\cc{B}} \subseteq \cc{TFNP}$ and $\cc{A}^* \subseteq \cc{TFNP}$.
\end{observation}

\subsection{The Choice of Complete Problem}
 To justify our definition, we now show that our choice for the complete problem for \cc{A} does not matter when defining $\cc{A}^{\cc{B}}$ as long as there exists a black-box reduction to/from that problem from/to the canonical complete problem for \cc{A} (up to \cc{FP^B} reductions). We note that \cref{thm: A_to_A'} relies crucially on the fact that solutions to $\prob{A}^\prob{B}$ are internally consistent and is the reason we insist on internal consistency in \cref{def: A^B}.
 
\begin{theorem}
    \label{thm: A_to_A'}
    Let $\problem{A_1}, \problem{A_2}$ and $\problem{B}$ be \cc{TFNP} problems and $\problem{A_1}, \problem{A_2}$ be circuit problems. If $\problem{A_1}$ is black-box many-one reducible to $\problem{A_2}$, then $\problem{A_1}^\problem{B}$ is black-box many-one reducible to $\problem{A_2}^\problem{B}$ under $\cc{FP}^\problem{B}$ reductions.
\end{theorem}

\begin{proof}
    Let $(f, g)$ be a pair of polynomial-time reduction algorithms such that on an $\problem{A_1}$ instance $C_1$, $f(C_1) = C_2$ generates an $\problem{A_2}$ instance $C_2$. And given a solution $\pi_2$ for $C_2$, $g(C_2, \pi_2) = \pi_1$ outputs a solution $\pi_1$ for $C_1$. In particular, we emphasize that $(f, g)$ are black-box reduction algorithms. I.e., they evaluate $C_1$ on polynomially many inputs and construct $C_2$ as a circuit with $C_1$-gates.

    We define a pair of $\cc{FP}^\problem{B}$ reduction algorithms $(f_1^\problem{B}, g_1^\problem{B})$. In particular, the reduction algorithms maintain a polynomial-sized table $T$ containing query-answer pairs to \problem{B}.

    On an $A_1^\problem{B}$ instance $C_1^{\problem{B}}$, $f_1^\problem{B}$ starts by instantiating an empty lookup table $T$. Next, it simulates $f$: whenever $f$ would evaluate $C_1^\problem{B}$ and hence query the oracle $B$ on some input $x$, it checks if $T$ contains the query-answer pair. If yes, it returns the answer stored in $T$. Otherwise, it uses its $\problem{B}$ oracle to obtain a query-answer pair, stores it in $T$ and returns the answer. At the end of its simulation of $f$, it constructs a circuit $C'$ (with $C_1^\problem{B}$-gates). It further modifies $C'$ as follows: the table $T$ is hardwired into $C'$ and any $\problem{B}$ gate in $C'$ is modified to always look for an answer from $T$ if available. The modified circuit $C_2^\problem{B}$ would be the output of $f_1^\problem{B}$.

    Given a solution $(\pi_2, w_2)$ to $C_2^\problem{B}$, $g_1^\problem{B}$ starts by simulating the verifier for $\problem{A}_2^\problem{B}$ and adding the query-answer pairs from $w_2$ to the table $T$. Note that due to the additional modification in $C_2^\problem{B}$, there will be no inconsistency between $w_2$ and $T$. Next, it simulates $g$ and similarly maintains the table $T$ in the same manner as $f_1^\problem{B}$. At the end of the simulation, it outputs $\pi_1$ and oracle answers $w_1$ for verifying the solution. 

    Internal consistency follows from the use of table $T$ and it remains to show the correctness of the reduction.

    Consider any function $b$ that is a restriction of $\problem{B}$ and is consistent with $T$. Since $b$ is a function, we may view $C_1^b$ as a vanilla circuit by hardwiring the truth table of $b$. Let $C_2^b:=f(C_1^b)$. Since $f_1^\problem{B}$ simulates $f$ and $b$ is consistent with $T$, $C_2^b$ is the same as $C_2^\problem{B}$ with the oracle gates switched and hence $\pi_2$ is a solution to $C_2^b$. Since $g_1^\problem{B}$ simulates $g$ and by correctness of $(f,g)$, $\pi_1$ is a solution to $C_1^b$. Therefore, $\pi_1$ with the associated oracle answers $w_1$ must be a solution for $C_1^\problem{B}$ since $C_1^b$ and $C_1^\problem{B}$ have the same behavior with respect to the verifier. 
\end{proof}

\begin{corollary}
    \label{cor: A_to_A'}
    Let $\problem{A_1}, \problem{A_2}$, $\problem{B}$, and $\problem{C}$ be \cc{TFNP} problems. If all the following hold, then $\problem{A}_1^{\problem{B}}$ reduces to $\prob{C}$ under many-one reductions.
    \begin{enumerate}
        \item $\problem{A_1}$ is black-box many-one reducible to $\problem{A_2}$.
        \item $\problem{A}_2^{\problem{B}}$ has a Turing reduction to \prob{C}.
        \item $\prob{B}$ has a Turing reduction to $\prob{C}$.
        \item $\prob{C}$ is Turing-closed.
    \end{enumerate}
\end{corollary}
\begin{proof}
    By \cref{thm: A_to_A'}, $\problem{A}_1^{\problem{B}}$ reduces to $\problem{A}_2^{\problem{B}}$ under $\cc{FP}^{\problem{B}}$ reductions. Since $\prob{B}$ has a Turing reduction to $\prob{C}$, $\problem{A}_1^{\problem{B}}$ reduces to $\problem{A_2}^{\problem{B}}$ under $\cc{FP}^{\problem{C}}$ reductions. This combined with the fact that $\problem{A}_2^{\problem{B}}$ has a Turing reduction to $\problem{C}$ implies $\problem{A}_1^{\problem{B}}$ reduces to $\problem{C}$ under $\cc{FP}^{\problem{C}}$ reductions. Since $\problem{C}$ is Turing-closed, this implies that $\problem{A}_1^{\problem{B}}$ reduces to $\problem{C}$ under many-one reductions.
\end{proof}

\begin{observation}
    \label{obv: ppp_not_selflow}
    We note that $\cc{PPP}$ is likely not self-low since even evaluating the input circuit to $\cc{PPP^{PPP}}$ is an $\cc{FP^{PPP}}$ computation, and we do not believe $\cc{FP^{PPP}} = \cc{PPP}$ \cite{fleming2024black}.
\end{observation}



 We now show that our choice for the complete problem for \cc{B} also does not matter as long as there exists a reduction to/from that problem from/to the canonical complete problem for \cc{B}.
\begin{theorem}
    \label{thm: B_to_B'}
    Let $\problem{B_1}, \problem{B_2}$ and $\problem{A}$ be \cc{TFNP} problems. If $\problem{B_1}$ is many-one reducible to $\problem{B_2}$, then $\problem{A}^\problem{B_1}$ is many-one reducible to $\problem{A}^\problem{B_2}$.
\end{theorem}
\begin{proof}
    Let $(f, g)$ be the many-to-one reduction from \problem{B_1} to \problem{B_2}. We now show a reduction from $\problem{A}^\problem{B_1}$ to $\problem{A}^\problem{B_2}$. Let $C^{\prob{B}_1}$ be the input to $\problem{A}^\problem{B_1}$. Let $G$ be a circuit with a $\prob{B}_2$ oracle gate defined as follows. $G$ on input $x$, computes $f(x)$, feeds $f(x)$ into a \problem{B_2}-gate, and applies $g$ to the output of its \problem{B_2}-gate. 

    On an input $C_1^\problem{B_1}$, the reduction from $\problem{A}^\problem{B_1}$ to $\problem{A}^\problem{B_2}$ simply replaces each \problem{B_1} gate with $G$ and outputs the resulting circuit $C_2^\problem{B_2}$. Given $(\pi, w_1, \ldots, w_t)$ as a solution for $C_2^\problem{B_2}$, the reduction outputs $(\pi, g(w_1), \ldots, g(w_t))$ as a solution.
    
    By the correctness of the reduction $(f,g)$, $g(w_i)$ are proper solutions to the $\problem{B_1}$ oracle queries. Furthermore, by our construction of $C_2^\problem{B_2}$, any evaluation on $C_1^\problem{B_1}$ assisted by $g(w_i)$ would be exactly the same as evaluation on $C_2^\problem{B_2}$ assisted by $w_i$. As such, $(\pi, g(w_1), \ldots, g(w_t))$ is a valid solution to $C_1^\problem{B_1}$.

    It remains to verify that internal consistency is preserved. Assume, for the sake of contradiction, that $g(w_i) \neq g(w_j)$ are solutions to the same \problem{B_1} query $x$ and hence $w_i \neq w_j$. However, $w_i$ and $w_j$ are solutions to the same \problem{B_2} query $f(x)$ in $C_2^\problem{B_2}$. This contradicts the internal consistency of the solution $(\pi, w_1, \ldots, w_t)$. 

\end{proof}

\subsection{The Choice of Complete Problem (Continued)}
\label{subsec: complete_prob_choice}

We have shown that for the purpose of defining the complexity class $\cc{A}^\cc{B}$, the choice of complete problem chosen for $\cc{A}$ does not matter under $\cc{FP}^{\cc{B}}$ reductions (\cref{def: A^B_classes}, \cref{thm: A_to_A'}). \emph{However}, it may matter if we are only considering many-one reductions rather than $\cc{FP}^{\cc{B}}$ reductions. This makes the definition of $\cc{A}^{\cc{B}}$ somewhat less robust than one might like.

However, we observe that this distinction is irrelevant for self-low classes since for any Turing-closed problem $\prob{C}$, $\cc{FP}^{\prob{C}}$ reductions to $\prob{C}$ and many-one reductions to $\prob{C}$ are equivalent.
\begin{observation}
    Let $\prob{C}, \prob{D}$ be $\cc{TFNP}$ circuit problems that are black-box many-one reducible to each other. Let $\prob{C}^1 = \prob{C}$ and $\prob{C}^i = \prob{C}^{(\prob{C}^{i-1})}$, similarly for $\prob{D}$. If $\prob{C}$ is Turing-closed and  ${\prob{C}}^{\prob{C}}$ has a many-one reduction to $\prob{C}$, then for all positive integers $i$, $\prob{C}^i$ is many-one equivalent to $\prob{D}^i$ under black-box reductions.
\end{observation}
Since we show that $\cc{PPA}$, $\cc{PLS}$, and $\cc{LOSSY}$ are self-low (\cref{thm: ppa_self_low}, \cref{thm: pls_self_low}, \cref{lem: lossy_self_low}), the choice of complete problem for these classes does not matter at all, at least from the perspective of hierarchies.




\subsection{Is Internal Consistency Necessary: an Example}\label{sec:example}
\snote{Rewrote the following paragraph.}
We now present an example of what goes wrong if we do not enforce internal consistency in \cref{def: A^B}: lifting of reductions as done in \cref{thm: A_to_A'} breaks down. We will define two $\TFNP$ problems $\problem{P}$ and $\problem{Q}$. $\problem{P}$ will reduce to $\problem{Q}$ via a many-to-one reduction, but without internal consistency we will not be able to lift this to a reduction from $\problem{P^S}$ to $\problem{Q^S}$ for some $\problem{S}\in \TFNP$. 

\begin{definition}
    The problem $\problem{P}$ is defined as follows. Given a circuit $C:[N] \rightarrow [N/10]$, output one of the following.
    \begin{enumerate}
        \item[(1)] Two collisions of distinct values: distinct $x_1, x_2, y_1, y_2 \in [N]$ such that $C(x_1) = C(x_2) \neq C(y_1) = C(y_2)$.
        \item[(2)] A 5-wise collision: distinct $x_1, \ldots, x_5 \in [N]$ such that $C(x_1) = C(x_2) = \ldots = C(x_5)$. 
    \end{enumerate}
\end{definition}

\begin{definition}
    The problem $\problem{Q}$ is defined as follows. Given a circuit $D:[N + 3] \rightarrow [N/10]$, output one of the following.
    \begin{enumerate}
        \item[(1)] Two 3-wise collisions of distinct values: distinct $x_1, x_2, x_3, y_1, y_2, y_3 \in [N+3]$ such that $D(x_1) = D(x_2) = D(x_3) \neq D(y_1) = D(y_2) = D(y_3)$.
        \item[(2)] An 8-wise collision: distinct $x_1, \ldots, x_8 \in [N+3]$ such that $D(x_1) = D(x_2) = \ldots = D(x_8)$. 
    \end{enumerate}
\end{definition}

It is not hard to verify that both problems are in $\TFNP$. We note that in both problems, type-(1) solutions are the more interesting cases and type-(2) solutions are meant to ensure that the problems are total. 

We now present a simple many-to-one reduction $(f,g)$ from $\problem{P}$ to $\problem{Q}$.

\begin{theorem}
    $\problem{P}$ is many-to-one reducible to $\problem{Q}$.
\end{theorem}
\begin{proof}
    Given an instance $C:[N] \rightarrow [N/10]$ for $\problem{P}$, the reduction constructs $D := f(C)$ as $D(x) := C(x \mod N)$.

    Next, we describe the second half of the reduction: the transformation $g$ on a solution of $\problem{Q}(D)$. 
    \begin{enumerate}
        \item[(1)] Suppose $g$ receives a type-(1) solution $(x_1, x_2, x_3, y_1, y_2, y_3)$. Let $x_i' = x_i \mod N$ and $y_i' = y_i \mod N$ for $i \in \{1,2,3\}$. Since $x_1, x_2, x_3$ are distinct, there must exist $x_i' \neq x_j'$. Moreover $C(x_i') = D(x_i) = D(x_j) = C(x_j')$. By the exact same argument, $g$ finds another colliding pair $y_k' \neq y_\ell'$ and $C(y_k') = C(y_\ell')$. The reduction outputs $x'_i, x'_j, y'_k, y'_\ell$ as a type-(1) solution. \snote{Added previous sentence.}
        
        We note that for any $i,j \in \{1,2,3\}$, we have $x_i' \neq y_j'$ since $C(x_i')=D(x_i) \neq D(y_j) = C(y_j')$. We will see that this guarantee does not hold in the absence of internal consistency when relative to an oracle $S$.
        \item[(2)] Suppose $g$ receives a type-(2) solution $(x_1, \ldots, x_8)$. It removes any $x_i \in \{N, N+1, N+2\}$ and outputs any five of the remaining $x_i$ as a type-(2) solution. Note that at most 3 items are removed and each remaining $x_i$ satisfies $C(x_i) = D(x_i)$. 
    \end{enumerate}
\end{proof}

We claim that the above reduction fails to reduce $\problem{P^S}$ to $\problem{Q^S}$, in the following scenario. For simplicity assume that the circuit $C^\problem{S}$ has exactly one oracle gate. Now a type-(1) solution to $\problem{Q^S}$ may look like $D(0, w_0) = D(1, w_1) = D(2, w_2) \neq D(N, w_0') = D(N+1, w_1') = D(N+2, w_2')$. Directly applying $g$ leads to a solution looking like
\[ C(0,w_0) = C(1,w_1) \neq C(1, w_1') = C(2, w_2') \; . \]
Notice that this violates the principle that $C^\problem{S}$ should behave like a function. We emphasize that we are not claiming that there is no reduction from $\problem{P}$ to $\problem{Q}$ which directly lifts to a reduction from $\problem{P^S}$ to $\problem{Q^S}$ for all $\problem{S} \in \cc{TFNP}$, merely that without internal consistency, we cannot guarantee that all such reductions lift. \snote{Added previous sentence.}

%% file: ppa.tex
We show that \cc{PPA} is self-low.


\snote{Changed from Turing reduction to many-one reduction since our proof trivially supplies that.}
\begin{lemma}
    \label{lem: lonely^lonely}
    $\prob{Lonely}^{\prob{Lonely}}$ reduces to \prob{Lonely} under black-box many-one reductions.
\end{lemma}
\begin{proof}
    Recall that we work with the definition of $\prob{Lonely}$ that consists of $N$ vertices where $N$ is odd. Moreover, for the oracle queries, we assume that each query is specified by a string in $Q := \bit^{m_1} \setminus \{0\}$ and each solution to a query is specified by a string in $W := \bit^{m_2} \setminus \{0\}$. And the base problem takes in strings from $S :=  \bit^{n} \setminus \{0\}$ and outputs strings in $S$. All these assumptions are simply a matter of syntactic encoding choices and do not affect the final result. We therefore assume without loss of generality that the input to the reduction is a circuit $C^{\prob{Lonely}}: S \rightarrow S$. The vertex set for the $\prob{Lonely}$ instance to which we reduce is then
    \[ S \times \prod_{i = 1}^{2t} (Q \times W) \; ,\]
    which has an odd number of vertices. 
    
    Next we construct the circuit $C'$, which has a very delicate behavior on input $(x_1, q_1, w_1, \ldots, q_{2t}, w_{2t})$. Let $x_2 = C^*(x_1, w_1, \ldots, w_t)$. $C'$ simulates $C^*(x_1, w_1, \dots, w_t)$ and then $C^*(x_2, w_{t+1}, \dots, w_{2t})$ as follows. Suppose that the simulation is about to make the $i$-th (where $i$ starts from 1) oracle query (which will purportedly be answered by $w_i$). $C'$ checks the following in order.
    \begin{enumerate}
        \item 
        \label{item: first}
        It checks if the query $q = q_i$. If not, $C'$ returns $(x_1, q_1, w_1, \ldots, q_i \xor q, w_i, \ldots, q_{2t}, w_{2t})$.
        \item 
        \label{item: second}
        It then scans for a minimum $i' < i$ such that $q_{i} = q_{i'}$ but $w_i \neq w_{i'}$. If it finds one, $C'$ returns $(x_1, q_1, w_1, \dots, q_i, w_i \oplus w_{i'}, \dots, q_{2t}, w_{2t})$.
        \item 
        \label{item: third}
        It then checks if $w_i$ is a solution to $q_i$. If not,  $C'$ returns $(x_1, q_1, w_1, \ldots, q_i, C_i(w_i), \ldots, q_{2t}, w_{2t})$, where $C_i$ is the input circuit encoded in $q_i$.
    \end{enumerate}
    If \cref{item: first}, \cref{item: second}, or \cref{item: third} is triggered at iteration $i$, we say condition 1, 2, or 3 was triggered on iteration $i$ respectively. If the whole simulation finishes without any of the above checks being triggered, $C'$ outputs $(x_2, q_{t+1}, w_{t+1}, \ldots, q_{2t}, w_{2t}, q_1, w_1, \ldots, q_{t}, w_{t})$.

    The reduction feeds $C'$ to its $\prob{Lonely}$ oracle, and receives back an answer $(x_1, q_1, w_1, \dots, q_{2t}, w_{2t})$. If this is a type 1 solution to $\prob{Lonely}$ (i.e., $C'(x_1, q_1, w_1, \dots, q_{2t}, w_{2t}) = (x_1, q_1, w_1, \dots, q_{2t}, w_{2t})$), the reduction outputs $(x_1, w_1, \dots, w_t)$. Otherwise, the reduction outputs $(x_1, w_1, \dots, w_{2t})$.

    The reduction clearly runs in polynomial time, so we proceed to correctness. First, observe that if condition 2 is triggered at iteration $i$, then all earlier occurrences of the query $q_i$ have the same recorded answer $w_{i'}$. Otherwise condition 2 would have been triggered even earlier. Next, if two inputs to $C'$ share the same prefix up till the pair $(q_i, w_i)$, then the executions of $C'$ on both inputs are identical at every iteration $j < i$ and make the same query $q$ to the $i$-th oracle gate. We call this property prefix invariance.

    Consider running $C'(x_1, q_1, w_1, \dots, q_{2t}, w_{2t})$. We now show correctness of the reduction by case analysis on how $C'(x_1, q_1, w_1, \dots, q_{2t}, w_{2t})$ is evaluated.
    \begin{enumerate}
        \item The evaluation of $C'(x_1, q_1, w_1, \dots, q_{2t}, w_{2t})$ triggers condition $1$ on evaluation $i$.
        \[ C'(C'(x_1, q_1, w_1, \dots, q_{2t}, w_{2t}))
        = C'(x_1, q_1, w_1, \dots, q_i \oplus q, w_i, \dots, q_{2t}, w_{2t})\]
        Let $\hat{q}$ be the value of $q$ during iteration $i$ of the evaluation of $C'(x_1, q_1, w_1, \dots, q_i \oplus q, w_i, \dots, q_{2t}, w_{2t})$. By the  prefix invariance of $C'$, $\hat{q} = q$. But since $q \neq q_i$ and $q_i \neq 0$, $\hat{q} \neq q_i \oplus \hat{q}$. Therefore,
        \begin{align*}
            C'(x_1, q_1, w_1, \dots, q_i \oplus q, w_i, \dots, q_{2t}, w_{2t})
        &= (x_1, q_1, w_1, \dots, q_i \oplus q \oplus q, w_i, \dots, q_{2t}, w_{2t})\\
        &= (x_1, q_1, w_1, \dots, q_{2t}, w_{2t}).
        \end{align*}
        This implies that $(x_1, q_1, w_1, \dots, q_{2t}, w_{2t})$ is not a solution to  $\prob{Lonely}$ on $C'$. Therefore, this case cannot happen.
        \item The evaluation of $C'(x_1, q_1, w_1, \dots, q_{2t}, w_{2t})$ triggers condition $2$ on evaluation $i$.
        \[C'(C'(x_1, q_1, w_1, \dots, q_{2t}, w_{2t}))
            = C'(x_1, q_1, w_1, \dots, q_i, w_i \oplus w_{i'}, \dots, q_{2t}, w_{2t}) \]
        Consider the evaluation of $C'(x_1, q_1, w_1, \dots, q_i, w_i \oplus w_{i'}, \dots, q_{2t}, w_{2t})$. Again, by prefix invariance, no condition is triggered during any iteration $j < i$. 
        Moreover, $w_i \oplus w_{i'} \neq 0$ because $w_i \neq w_{i'}$, and $w_i \oplus w_{i'} \neq w_{i'}$ because $w_i \neq 0$. Hence condition 2 is triggered again at iteration $i$, and the answer at the minimum qualifying index is $w_{i'}$. Therefore,
        \begin{align*}
            C'(x_1, q_1, w_1, \dots, q_i, w_i \oplus w_{i'}, \dots, q_{2t}, w_{2t})
            &= (x_1, q_1, w_1, \dots, q_i, w_i \oplus w_{i'} \oplus w_{i'}, \dots, q_{2t}, w_{2t})\\
            &= (x_1, q_1, w_1, \dots, q_{2t}, w_{2t}).
        \end{align*}
        This implies that $(x_1, q_1, w_1, \dots, q_{2t}, w_{2t})$ is not a solution to  $\prob{Lonely}$ on $C'$. Therefore, this case cannot happen.
        \item The evaluation of $C'(x_1, q_1, w_1, \dots, q_{2t}, w_{2t})$ triggers condition $3$ on iteration $i$. Note that this implies that $C_i(w_i) \neq w_i$ and $C_i(C_i(w_i)) = w_i$.
        \[C'(C'(x_1, q_1, w_1, \dots, q_{2t}, w_{2t}))
            = C'(x_1, q_1, w_1, \ldots, q_i, C_i(w_i), \ldots, q_{2t}, w_{2t}) \]
        
        Similar reasoning as in the previous case tells us that the evaluation of $C'(x_1, q_1, w_1, \ldots, q_i, C_i(w_i), \ldots, q_{2t}, w_{2t})$ will not trigger any  condition \cref{item: first}, \cref{item: second}, or \cref{item: third} on any iteration $j < i$ and will not trigger \cref{item: first} on iteration $i$. 
        
        We now show that it will also not trigger \cref{item: second} on iteration $i$. Say for the sake of contradiction that it did. This would imply that there is some $i' < i$ such that $q_{i'} = q_i$. But since \cref{item: second} did not trigger at iteration $i$ of the execution of $C'(x_1, q_1, w_1, \dots, q_{2t}, w_{2t})$, we must have that $w_{i'} = w_i$. Since that execution passed iteration $i'$, $w_{i'}$ must be a valid solution to the \prob{Lonely} instance $q_{i'}$. But since $q_i = q_{i'}$ and $w_i = w_{i'}$, $w_i$ is a solution to $q_i$, which contradicts the fact that the evaluation of $C'(x_1, q_1, w_1, \dots, q_{2t}, w_{2t})$ triggers condition $3$ on iteration $i$.
        
        Therefore, the evaluation of $C'(x_1, q_1, w_1, \ldots, q_i, C_i(w_i), \ldots, q_{2t}, w_{2t})$ will check if $C_i(w_i)$ is a solution to $q_i$. It will find that it is not since $C_i(C_i(w_i)) \neq C_i(w_i)$ and $C_i(C_i(C_i(w_i))) = C_i(w_i)$, which both follow from the fact that $w_i$ is not a solution to $q_i$.
        \begin{align*}
            C'(x_1, q_1, w_1, \ldots, q_i, C_i(w_i), \ldots, q_{2t}, w_{2t})
            &= (x_1, q_1, w_1, \ldots, q_i, C_i(C_i(w_i)), \ldots, q_{2t}, w_{2t})\\
            &= (x_1, q_1, w_1, \dots, q_{2t}, w_{2t}).
        \end{align*}
        This implies that $(x_1, q_1, w_1, \dots, q_{2t}, w_{2t})$ is not a solution to $\prob{Lonely}$ on $C'$. Therefore, this case cannot happen.
        \item The whole simulation of $C'(x_1, q_1, w_1, \dots, q_{2t}, w_{2t})$ finishes without triggering any of \cref{item: first}, \cref{item: second}, or \cref{item: third}. Let $x_3 = C^*(x_2, w_{t+1}, \dots, w_{2t})$. We distinguish whether $x_1 = x_2$. If $x_1 = x_2$, then the reduction clearly succeeds, and it outputs consistent subwitnesses because neither \cref{item: first} nor \cref{item: second} is ever triggered during the evaluation of $C'(x_1, q_1, w_1, \dots, q_{2t}, w_{2t})$. Now suppose that $x_1 \neq x_2$. We show that $x_1 \neq x_3$ by contradiction. Say $x_1 = x_3$. Then $(q_1, w_1, \dots, q_{2t}, w_{2t})$ is internally consistent, the simulation $C^*(x_1, w_1, \dots, w_{t}) = x_2$ makes oracle gate queries $q_i$ which can be answered with subwitnesses $w_i$, and the simulation $C^*(x_2, w_{t+1}, \dots, w_{2t}) = x_1$ makes oracle gate queries $q_{t+i}$ which can be answered with subwitnesses $w_{t+i}$. Consequently, after rotating the two transcripts, both simulations again pass all three checks. By the definition of $C'$, we therefore have
        \begin{align*}
            C'(C'(x_1, q_1, w_1, \dots, q_{2t}, w_{2t}))
            &= C'(x_2 , q_{t+1}, w_{t+1}, \dots, q_{2t}, w_{2t}, q_1, w_1, \dots, q_t, w_t)\\
            &= (x_1, q_1, w_1, \dots, q_{2t}, w_{2t}).
        \end{align*}
        Since $x_1 \neq x_2$, the displayed equality shows that applying $C'$ twice returns to the original vertex. Hence this vertex is not a valid solution, a contradiction. So $x_1 \neq x_3$. Therefore, $(x_1, w_1, \dots, w_{2t})$ is a solution to $\prob{Lonely}^{\prob{Lonely}}$. Internal consistency is handled by the fact that neither \cref{item: first} nor \cref{item: second} is ever triggered during the evaluation of $C'(x_1, q_1, w_1, \dots, q_{2t}, w_{2t})$.
    \end{enumerate}
\end{proof}
\ppa
\begin{proof}
    $\cc{PPA}$ is trivially contained in $\cc{PPA}^*$.
    

     \cref{lem: lonely^lonely} lets us conclude that $\prob{Lonely}^{\prob{Lonely}}$ has a many-one reduction to $\prob{Lonely}$. Therefore, $\cc{PPA}^{\cc{PPA}} = \cc{PPA}$. This immediately collapses the whole $\cc{PPA}^*$ hierarchy to $\cc{PPA}$ by induction.
    
    
\end{proof}

%% file: pls2.tex
We now show that \cc{PLS} is self-low. To do so, we work with the established \cc{PLS}-complete problem \prob{Iter} (introduced in \cite{BussKrajicek94} and \cite{moriokaclassification}), \prob{IterWithSource}, and a variant we call \prob{Iter2}.

\subsubsection{The $\cc{PLS}$-Complete Problems}
To ensure syntactic totality of $\prob{Iter}$, \prob{IterWithSource}, and $\prob{Iter2}$, we say that 0 is a solution if $S$ is the identity on the start vertex.
\begin{definition}
    The problem \prob{Iter} is defined as follows. The input is $S: [2^n] \rightarrow [2^n]$. Output any of the following solutions.
    \begin{enumerate}
        \item $0$ if $S(0) = 0$,
        \item $x$ s.t. $S(x) > x$ and $S(S(x)) \leq S(x)$.
    \end{enumerate}
\end{definition}

\begin{definition}
    The problem \prob{IterWithSource} is defined as follows. The input is $S: [2^n] \rightarrow [2^n]$ and $\alpha \in [2^n]$. Output any of the following solutions.
    \begin{enumerate}
        \item $0$ if $S(\alpha) \leq \alpha$,
        \item $x$ s.t. $S(x) > x$ and $S(S(x)) \leq S(x)$.
    \end{enumerate}
\end{definition}

\begin{definition}
    The problem \prob{Iter2} is defined as follows. The input is $S: [2^n] \rightarrow [2^n]$. Output any of the following solutions.
    \begin{enumerate}
        \item $0$ if $S(0) = 0$,
        \item $x$ s.t. $S(x) < x$,
        \item $x$ s.t. $S(x) > x$ and $S(S(x)) \leq S(x)$.
    \end{enumerate}
\end{definition}

\begin{lemma}
    \prob{Iter2} is \cc{PLS}-complete.
\end{lemma}
\begin{proof}
    We first reduce \prob{Iter} to \prob{Iter2}. Given an instance $S$ of \prob{Iter}, we let $S'(x)$ be $x$ if $S(x) < x$ and $S(x)$ otherwise. We feed $S'$ to our \prob{Iter2} oracle to get back a solution $y$ which we output. If $y=0$ and $S'(0)=0$, then $S(0)=0$, so $y$ is a valid solution to the original \prob{Iter} instance. Otherwise, since $S'(x) \geq x$ for every $x$, the solution $y$ must satisfy $S'(y) > y$ and $S'(S'(y)) \leq S'(y)$. We then have $S(y) = S'(y)$ and $S'(S(y)) \leq S(y)$. This implies that $S(S(y)) \leq S(y)$, since by construction $S' \geq S$ for all inputs. Therefore, $y$ is a valid solution to \prob{Iter}.

    We now reduce \prob{Iter2} to \prob{Iter}. Given an instance $S$ of \prob{Iter2}, the reduction calls its \prob{Iter} oracle on $S$ to get back an answer $y$ which it outputs. If $y=0$ and $S(0)=0$, then $y$ is a valid solution to \prob{Iter2}. Otherwise, $S(y) > y$ and $S(S(y)) \leq S(y)$, so $y$ is also a valid solution to \prob{Iter2}.
\end{proof}

\subsubsection{The Reduction}
We are now ready to show the main reduction.
\snote{Changed it from reducing to $\prob{Iter}$ to reducing to $\prob{IterWithSource}$}
\begin{lemma}
    \label{lem: Iter_self_low}
    $\prob{Iter}^{\prob{Iter2}}$ reduces to $\prob{IterWithSource}$ under black-box Turing reductions.
\end{lemma}
\begin{proof}
    We show the reduction from $\prob{Iter}^{\prob{Iter2}}$ to $\prob{IterWithSource}$. Let $S$ be the input to the reduction. We assume that $S$ has $t$ oracle gates, where the $i^{\text{th}}$ oracle gate has an output of size $s_i$. Let $S^2$ be the circuit which simply composes $S$ with itself. Note that $S^2$ has $2t$ oracle gates.

    The reduction first computes $a_1, \dots, a_t$ such that $a_1, \dots, a_t$ are consistent for $S^*(0, a_1, \dots, a_t)$ and $S^*(0, a_1, \dots, a_t) \neq \bot$ by repeated calls to its $\prob{IterWithSource}$ oracle. If $S^*(0, a_1, \dots, a_t) = 0$, the reduction outputs $(0, a_1, \dots, a_t)$. Otherwise, let $y_0 = S^*(0, a_1, \dots, a_t) > 0$ and define
    \[ \alpha = (0, a_1, \dots, a_t, 0^{s_1}, \dots, 0^{s_t}). \]

    The reduction then constructs
    \[S': \{ 0, 1\}^n \times  \prod_{j=1}^t \{ 0, 1\}^{s_j} \times \prod_{j=1}^t \{ 0, 1\}^{s_j} \rightarrow \{ 0, 1\}^n \times \prod_{j=1}^t \{ 0, 1\}^{s_j} \times \prod_{j=1}^t \{ 0, 1\}^{s_j} \ .\]

    $S'(x, w_1, \dots, w_{2t})$ behaves as follows. Let $y = S^*(x, w_1, \dots, w_{t})$ and if $y \neq \bot$, let $z = S^*(y, w_{t+1}, \dots, w_{2t})$. The behavior of $S'$ is now split into four cases.
    \begin{enumerate}
        \item
        If $y = \bot$ or internal consistency is violated in $S^*(x, w_1, \dots, w_{t})$, output $(x, w_1, \dots, w_{2t})$.
        \item 
        If $y \neq \bot$, the evaluation of $S^*(x, w_1, \dots, w_{t})$ is internally consistent, and $y \leq x$, output $(x, w_1, \dots, w_{2t})$.
        \item 
        Consider when $y \neq \bot$, the evaluation of $S^*(x, w_1, \dots, w_{t})$ is internally consistent, $y > x$, and one of the following occurs: $z = \bot$ or the evaluation of ${S^{2}}^{*}(x, w_1, \dots, w_{2t})$ violates internal consistency. Let $m$ be the smallest index in $\{t+1, \ldots, 2t\}$ such that $w_m$ is not a solution to oracle gate query $m$ in the evaluation of ${S^{2}}^{*}(x, w_1, \dots, w_{2t})$ or $w_m$ violates internal consistency in the evaluation of ${S^{2}}^{*}(x, w_1, \dots, w_{2t})$ because gate query $m$ is the same as gate query $i$ for some $i < m$ and $w_m \neq w_i$. If $w_m$ violates internal consistency, output $(x, w_1, \dots, w_{m-1}, w_i, 0, \ldots, 0)$. Otherwise, $w_m$ is not a valid solution to an oracle gate query $u$ which encodes an \prob{Iter2} query $H: \{ 0, 1\}^q \rightarrow \{ 0, 1\}^q$. In this case, output $(x, w_1, \dots, w_{m-1}, H(w_m), 0, \dots, 0)$. 
        \item 
        Say $y \neq \bot$, $y > x$, $z \neq \bot$, and the evaluation of ${S^{2}}^{*}(x, w_1, \dots, w_{2t})$ is internally consistent. Output $(y, w_{t+1}, \dots, w_{2t}, 0, \dots, 0)$.
    \end{enumerate}

    The reduction calls its \prob{IterWithSource} oracle on $S', \alpha$ to get back an answer $(x, w_1, \dots, w_{2t})$ which would be the output of our reduction.

    \snote{Rewrote the following paragraph.}
    The reduction clearly runs in polynomial time since $S'$ translates to a polynomial size circuit and all other operations run in polynomial time. We now show the correctness. Note that if $S^*(0, a_1, \dots, a_t) = 0$, the reduction is clearly correct since it returns a valid type 1 solution to $\prob{Iter}^{\prob{Iter2}}$. We therefore assume that $S^*(0, a_1, \dots, a_t) \neq 0$. We next show that $S'(\alpha) > \alpha$, ensuring that the final \prob{IterWithSource} oracle call will only return type 2 solutions to \prob{IterWithSource}. Consider the evaluation of $S'(\alpha)$. Since by construction, $S^*(0, a_1, \dots, a_t) \neq \bot$ and the evaluation of $S^*(0, a_1, \dots, a_t)$ is internally consistent, case 1 is not triggered during the evaluation of $S'(\alpha)$. Furthermore, since by assumption, $y \neq 0$, case 2 is not triggered during the evaluation of $S'(\alpha)$. If $S'(\alpha)$ is evaluated using case 3, then $S'(\alpha) > \alpha$ since case 3 can only replace a zero value in $(0, a_1, \dots, a_t, 0, \dots, 0)$ with a nonzero value. Finally, if $S'(\alpha)$ is evaluated using case 4, then $S'(\alpha) > \alpha$ by $y > x$. In all cases, $S'(\alpha) > \alpha$.
    Let $v = (x, w_1, \dots, w_{2t})$ be the answer the final oracle call to $\prob{IterWithSource}$ returns. By the fact that $S'(\alpha) > \alpha$, we know that $v$ satisfies the following inequalities.
    \begin{equation}
        \label{eq: cond1}
        S'(v) > v
    \end{equation}
    \begin{equation}
        \label{eq: cond2}
        S'(S'(v)) \leq S'(v)
    \end{equation}

    We will divide our proof of correctness by which cases are used to evaluate $S'(v)$ and $S'(S'(v))$. Before we dive into the case analysis, we note that $S'$ evaluated using case 1 and case 2 is the identity and could not satisfy \cref{eq: cond1}. Hence $S'(v)$ has to be evaluated using either case 3 or case 4. Furthermore, since $S'$ evaluated as case 4 is strictly increasing, $S'(S'(v))$ could not be evaluated as case 4.
    
    \begin{enumerate}
        \item $S'(v)$ is evaluated using case 3 and $S'(S'(v))$ is evaluated using case 1 or case 2. This case cannot happen. Since $S'(x, w_1, \dots, w_{2t})$ was evaluated using case 3, $S^*(x, w_1, \dots, w_t) \neq \bot$ and its evaluation is internally consistent. Let $S'(x, w_1, \dots, w_{2t}) = (x, w_1', \dots, w_{2t}')$. Notice that $w_i = w_i'$ for all $i$ in $[1, t]$. Therefore, $S^*(x, w_1', \dots, w_t') \neq \bot$ and its evaluation is internally consistent. Therefore, $S'(x, w_1', \dots, w_{2t}')$ will not be evaluated using case 1 or case 2.  

        \item $S'(v)$ is evaluated using case 3 and $S'(S'(v))$ is evaluated using case 3. This case cannot happen.
        Let $m$ be the identified index in $S'(x, w_1, \dots, w_{2t})$ and $m'$ be the identified index in $S'(S'(x, w_1, \dots, w_{2t}))$. Note that $m \leq m'$. Consider first when $m < m'$. Then $S'(x, w_1, \dots, w_{2t}) = (x, w_1, \dots, w_{m-1}, \hat{w}_m, 0, \dots, 0)$ and $S'(x, w_1, \dots, w_{m-1}, \hat{w}_m, 0, \dots, 0) = (x, w_1, \dots, w_{m-1}, \hat{w}_m, 0, \dots, 0, \hat{w}_{m'}, 0, \dots, 0)$ for some $\hat{w}_m$ and $\hat{w}_{m'} \neq 0$. Therefore, $S'(S'(x, w_1, \dots, w_{2t})) > S'(x, w_1, \dots, w_{2t})$, violating \cref{eq: cond2}. If $m = m'$, this can only happen when $w_m$ and $H(w_m)$ are not solutions to $u$. 
        $S'(x, w_1, \dots, w_{2t}) = (x, w_1, \dots, w_{m-1}, H(w_m), 0, \dots, 0)$ and $S'(S'(x, w_1, \dots, w_{2t})) = (x, w_1, \dots, w_{m-1}, H(H(w_m)), 0, \dots, 0)$. Since $w_m$ is not a solution to \prob{Iter2} on instance $H$ and \cref{eq: cond1} implies $H(w_m) > w_m$, we must have $H(H(w_m)) > H(w_m)$. Therefore, $S'(S'(x, w_1, \dots, w_{2t})) > S'(x, w_1, \dots, w_{2t})$, contradicting \cref{eq: cond2}.


        \item $S'(v)$ is evaluated using case 4 and $S'(S'(v))$ is evaluated using case 1. This case cannot happen. In particular, since $S'(x, w_1, \dots, w_{2t})$ is evaluated using case 4, we know $S^*(y, w_{t+1}, \dots, w_{2t}) \neq \bot$ and is internally consistent.

        \item $S'(v)$ is evaluated using case 4 and $S'(S'(v))$ is evaluated using case 2. Notice that $S'(x, w_1, \dots, w_{2t}) = (y, w_{t+1}, \dots, w_{2t}, 0, \dots, 0)$ where $S^*(y, w_{t+1}, \dots, w_{2t}) = z$. Therefore, $S'(S'(x, w_1, \dots, w_{2t}))$ being evaluated using case 2 means that $z \leq y$. The reduction therefore outputs a valid solution in this case. 

        \item $S'(v)$ is evaluated using case 4 and $S'(S'(v))$ is evaluated using case 3. This case cannot happen. Let $S'(x, w_1, \dots, w_{2t}) = (y, w_{t+1}, \dots, w_{2t}, 0, \dots, 0)$. 
        Evaluation by case 3 implies that $S'(y, w_{t+1}, \dots, w_{2t}, 0, \dots, 0) = (y, w_{t+1}, \dots, w_{2t}, 0, \dots, 0, \hat{w}_m, 0, \dots)$. As such, $S'(S'(x, w_1, \dots, w_{2t})) > S'(x, w_1, \dots, w_{2t})$, contradicting \cref{eq: cond2}.

    \end{enumerate}
\end{proof}

\pls
\begin{proof}
    $\cc{PLS}$ is trivially contained in $\cc{PLS}^*$.

    We now show the other direction. Note that all the following conditions are satisfied.
    \begin{enumerate}
        \item \prob{Sink-of-DAG} has a black-box many-one reduction to \prob{Iter}. 
        \item $\prob{Iter}^\prob{Iter2}$ has a Turing reduction to \prob{IterWithSource} by \cref{lem: Iter_self_low}.
        \item \prob{Iter2} is many-one reducible to \prob{Iter}.
        \item \prob{IterWithSource} is Turing-closed since \cc{PLS} is Turing-closed.
    \end{enumerate}
    Therefore, \cref{cor: A_to_A'} tells us $\prob{Sink-of-DAG}^{\prob{Iter2}}$ has a many-one reduction to \prob{IterWithSource}, which has a many-one reduction to \prob{Sink-of-DAG}. Finally, by \cref{thm: B_to_B'}, $\prob{Sink-of-DAG}^{\prob{Sink-of-DAG}}$ has a many-one reduction to $\prob{Sink-of-DAG}^{\prob{Iter2}}$. Chaining these reductions gives us a reduction from $\prob{Sink-of-DAG}^{\prob{Sink-of-DAG}}$ to \prob{Sink-of-DAG}. Therefore, $\cc{PLS}^{\cc{PLS}} = \cc{PLS}$. This immediately collapses the whole $\cc{PLS}^*$ hierarchy to $\cc{PLS}$ by induction.

    
\end{proof}

%% file: lossy.tex
We now show that \cc{LOSSY} is self-low.

\begin{lemma}
    \label{lem: lossy_self_low}
    $\prob{Lossy}^{\prob{Lossy}}$ reduces to $(n-1)$-$\prob{Lossy}$ under many-one reductions.
\end{lemma}
\begin{proof}
    Let $C: \{ 0, 1\}^n \rightarrow \{ 0, 1\}^{n/2}, D: \{ 0, 1\}^{n/2} \rightarrow \{ 0, 1\}^n$ be the circuits that act as input to our $\prob{Lossy}^{\prob{Lossy}}$ problem. Without loss of generality, assume that $C, D$ collectively have $t$ \prob{Lossy} gates, and for some fixed polynomial $s(n) \geq 100 \log(t+100)$, every input to a $\prob{Lossy}$ oracle gate takes the form $c_i: \{ 0, 1\}^{s} \rightarrow \{ 0, 1\}^{s/2}, d_i: \{ 0, 1\}^{s/2} \rightarrow \{ 0, 1\}^{s}$. This can be achieved by padding and applying \cref{lem: lossy-equiv}.
    

    
    We now construct 
    \[ C': \{ 0, 1\}^n \times \{ 0, 1\}^{s \cdot t} \rightarrow \{ 0, 1\}^{n + s \cdot t - 1}, \]

    \[ D':  \{ 0, 1\}^{n + s \cdot t - 1} \rightarrow \{ 0, 1\}^n \times \{ 0, 1\}^{s \cdot t} \ . \]

    Let $N:= n + s \cdot t$. Let $D \circ C: \{ 0, 1\}^n \rightarrow \{ 0, 1\}^n$ be the composed circuit which consists of $C$ followed by $D$. Let $x_2 = C_*(x_1, w_1, \dots, w_t)$ and $x_3 = (D\circ C)_*(x_1, w_1, \dots, w_t)$.
    
    $C'(x_1, w_1, \dots, w_{t})$ computes $x_3 = (D\circ C)_*(x_1, w_1, \dots, w_t)$. If $x_3$ is an error output $(m, z)$ where $z$ encodes a \prob{Lossy} instance $c: \{ 0, 1\}^{s} \rightarrow \{ 0, 1\}^{s/2}, d: \{ 0, 1\}^{s/2} \rightarrow \{ 0, 1\}^{s}$, $C'$ outputs $(1, m, x_1, \dots, w_{m-1}, c(w_m), w_{m+1} \dots, w_{t}, 0, \dots, 0) \in \{ 0, 1\}^{N-1}$. Otherwise, $C'$ outputs $(0, x_2, w_1, \dots, w_{t}, 0, \dots, 0) \in \{ 0, 1\}^{N-1}$ where $x_2 = C_*(x_1, w_1, \dots, w_t)$.
    

    Next, we now define $D'$. We first consider the case when the input to $D'$ has the form $(0, x_2, w_1, \dots, w_{t}, 0, \dots, 0)$. $D'$ computes $x_3 = D_*(x_2, w_1, \dots, w_t)$. If $x_3$ is an error output, $D'$ outputs $0$. Otherwise, $D'$ outputs $(x_3, w_1, \dots, w_t)$. When the input to $D'$ has the form $(1, m, x_1, \dots, w_{t}, 0, \dots, 0)$, $D'$ computes ${(D\circ C)_*}(x_1, w_1, \dots, w_{m-1}, 0^{s}, w_{m+1}, \dots,  w_{t})$. If this results in a result output, $D'$ outputs $0$. Say it results in an error output $(m', z)$ where $z$ encodes a \prob{Lossy} instance $c: \{ 0, 1\}^{s} \rightarrow \{ 0, 1\}^{s/2}, d: \{ 0, 1\}^{s/2} \rightarrow \{ 0, 1\}^{s}$. 
    $D'$ outputs $(x_1, \dots, w_{m-1}, d(w_m), w_{m+1}, \dots, w_{t})$.

    The reduction feeds $C', D'$ to its \prob{Lossy} oracle to get back $(v_1, w_1, \dots, w_{t})$. The reduction outputs $v_1$ as well as witnesses used to evaluate $(D \circ C)_*(v_1, w_1, \dots, w_{t})$ in order.

    The reduction clearly runs in polynomial time. Notice also that since $c$ compresses by at least $s/2 \geq 50 \log(t+100)$ bits and $m$ requires exactly $\log_2(t)$ bits to specify, $C'$ compresses by at least 1 bit. To show correctness, let $v_2 = C_*(v_1, w_1, \dots, w_t)$ and $v_3 = (D\circ C)_*(v_1, w_1, \dots, w_t)$. We consider two cases.

    \begin{enumerate}
        \item $v_3$ is a result output. Then $C'(v_1, w_1, \dots, w_{t}) = (0, v_2, w_1, \dots, w_{t}, 0, \dots, 0)$. Notice that by construction, $D'(0, v_2, w_1, \dots, w_{t}, 0, \dots, 0) = (v_3, w_1, \dots, w_{t})$. By assumption, $v_1 \neq v_3$. Therefore, $v_1$ as well as the witnesses among $w_1, \dots, w_{t}$ used to evaluate $D(C(v_1))$ are a solution to our $\prob{Lossy}^{\prob{Lossy}}$ instance. The witnesses are all consistent since the first valid $w_i$ among $w_1, \dots, w_{t}$ is used to evaluate the oracle gates at every step.
        
        \item $v_3$ is an error output $(m, z)$. This cannot happen. Say $z$ encodes a \prob{Lossy} instance $c: \{ 0, 1\}^{s} \rightarrow \{ 0, 1\}^{s/2}, d: \{ 0, 1\}^{s/2} \rightarrow \{ 0, 1\}^{s}$. Then $C'(v_1, w_1, \dots, w_{t}) = (1, m, v_1, \dots, c(w_m), \dots, w_{t}, 0, \dots, 0)$. Notice that since $m \notin M$ (recall \cref{def: another_eval}) for the evaluation of $v_3 = (D\circ C)_*(v_1, w_1, \dots, w_t)$, ${(D\circ C)_*}(v_1, w_1, \dots, w_{m-1}, 0^{s}, w_{m+1}, \dots, w_{t})$ should have the exact same behaviour as $(D\circ C)_*(v_1, w_1, \dots, w_t)$ and output $(m, z)$ and $w_m$ is not a solution to \prob{Lossy} on $c, d$. This is because $C_*$ always tries $0$ as a solution to an oracle gate. As such, $D'(1, m, v_1, w_1, \dots, c(w_m), \dots, w_{t}, 0, \dots, 0)$ evaluates to $(v_1, w_1, \dots,  d(c(w_m)), \dots, w_{t}) = (v_1, w_1, \dots, w_{t})$. Therefore, $(v_1, w_1, \dots, w_m, \dots, w_{t})$ is not a solution to our oracle call to \prob{Lossy} on $C', D'$.
    \end{enumerate}
\end{proof}

\lossy
\begin{proof}
    \prob{Lossy} is trivially contained in $\cc{LOSSY}^*$.

     $\prob{Lossy}^\prob{Lossy}$ has a many-one reduction to $(n-1)$-\prob{Lossy} by \cref{lem: lossy_self_low}, which reduces to \prob{Lossy} by \cref{lem: lossy-equiv}. Chaining the reductions tells us $\prob{Lossy}^\prob{Lossy}$ reduces to \prob{Lossy}. Therefore, $\cc{LOSSY}^{\cc{LOSSY}} = \cc{LOSSY}$. This immediately collapses the whole $\cc{LOSSY}^*$ hierarchy to $\cc{LOSSY}$ by induction.
    
\end{proof}

%% file: consequences.tex
In this section, we demonstrate the potential of our new definitions of \cc{TFNP} subclasses with \cc{TFNP} oracles for developing a better understanding of important computational problems. 

\subsection{Number Theory}
Notably, our result that \cc{PPA} is self-low provides a potential way to classify the problem of deterministically generating large primes inside \cc{TFNP}. We first define the necessary problems.

\begin{definition}[\prob{Weak-Bertrand}]
    Given a string $1^n$, output a $32n$ bit prime $p$ such that $p > 2^n$.
\end{definition}
We refer to this problem as \prob{Weak-Bertrand} since Bertrand's postulate tells us that there always exists a prime between $2^n$ and $2^{n+1}$. The problem \prob{Bertrand} would ask us to find such a prime. In \prob{Weak-Bertrand}, we are asking for a prime between $2^n$ and $2^{32n}$. We now review relevant results.


\begin{lemma}[\cite{jevrabek2016integer}]
    \label{lem: factoring-in-ppa}
    Under the Generalized Riemann Hypothesis, \prob{Factor} is in $\cc{PPA}$ and \cc{PPP}.
\end{lemma}


\begin{lemma}[\cite{korten2022derandomization}]
    \label{lem: korten-prime}
    \prob{Weak-Bertrand} reduces to $\prob{Lossy}^{\prob{Factor}}$ \footnote{\label{footnote: first} Technically, \cite{korten2022derandomization} showed this for \prob{Lossy} given access to an oracle which outputs all prime factors of a number rather than one non-trivial factor. We observe that one can obtain all factors of a number by simply applying $\prob{Factor}$ multiple times. Furthermore, Korten's reduction makes further $\prob{Factor}$ queries after finding an incompressible string in the $\prob{Lossy}^{\prob{Factor}}$ instance, but we note that these can trivially be encoded into the $\prob{Lossy}^{\prob{Factor}}$ instance.}.
\end{lemma}


We are able to leverage \cref{lem: factoring-in-ppa} and \cref{lem: korten-prime} to give a classification of \prob{Weak-Bertrand} into our newly defined classes. Furthermore, the fact that \cc{PPA} is self-low and \prob{Factor} is in \cc{PPA} (under the Generalized Riemann Hypothesis) suggests that \prob{Weak-Bertrand} (or even \prob{Bertrand}) may be easily reducible to a \cc{PPA}-complete problem.
\primes
\begin{proof}
    Both $\prob{Weak-Bertrand} \in \cc{LOSSY^{{PPA}}}$ and $\prob{Weak-Bertrand} \in \cc{LOSSY^{{PPP}}}$ follow directly by combining \cref{lem: korten-prime} with \cref{lem: factoring-in-ppa}. Observing that \prob{Lossy} has a trivial black-box reduction to the \cc{PPADS}-complete problem \prob{InjectivePigeon} then implies \prob{Weak-Bertrand} is in both $\cc{PPADS^{{PPA}}}$ and $\cc{PPADS^{{PPP}}}$. Similarly, the fact that $\prob{InjectivePigeon}$ has a trivial black-box reduction to $\prob{Pigeon}$ implies that \prob{Weak-Bertrand} is in $\cc{PPP}^{\cc{PPP}}$.
\end{proof}

\primesPPA
\begin{proof}
    Apply \cref{lem: factoring-in-ppa} and \cref{thm: ppa_self_low}.
\end{proof}

\snote{Removed an interpretation we had below since it only applied to black-box things.}